\documentclass{article}
\usepackage{ijcai19}

\usepackage{todonotes}
\usepackage{lineno}
\usepackage{xcolor}
\usepackage{amsmath}
\usepackage{proof}
\usepackage{amssymb}
\usepackage{amsthm}
\usepackage{ stmaryrd }
\usepackage{graphics}
\usepackage{graphicx}
\usepackage{amssymb}
\usepackage[hidelinks]{hyperref}
\hypersetup{
  pdfinfo={
    Title={Reasoning about Disclosure in Data Integration in the Presence of Source Constraints},
    Authors={Michael Benedikt, Pierre Bourhis, Louis Jachiet, Michaël Thomazo},
    Author={Michael Benedikt, Pierre Bourhis, Louis Jachiet, Michaël Thomazo},
    DOI={10.24963/ijcai.2019/215}
  }
}

\usepackage[all]{xy}

\usepackage{proof}
\usepackage{ifsym}
\usepackage{diagbox}
\usepackage{amsmath}
\usepackage{mathrsfs}
\usepackage{times}
\usepackage{soul}
\usepackage{url}
\usepackage{graphics}
\usepackage{graphicx}
\usepackage[utf8]{inputenc}
\usepackage[small]{caption}
\usepackage{graphicx}
\usepackage{amsmath}
\usepackage{booktabs}
\usepackage{xcolor}
\usepackage{paralist}
\usepackage{xspace}
\usepackage{algorithm}
\usepackage{algorithmic}
\usepackage{paralist}
\usepackage{amsthm}
\usepackage{tikz}
\usepackage{placeins}
\usepackage[all]{xy}
\usetikzlibrary{arrows}
\newcommand{\singleconstantrule}{\kw{SCEQrule}}

\newtheorem{example}{Example}
\newtheorem{lem}{Lemma}

\newtheorem{definition}{Definition}
\newtheorem{theorem}{Theorem}
\newtheorem{corollary}{Corollary}
\newtheorem{proposition}{Proposition}

\newcommand{\sigside}{{\cal S}_{\kw{Side}}}
\newcommand{\myeat}[1]{}
\newcommand{\sch}{\mathcal{S}}

\newcommand{\simplerewritingconstraints}{\Sigma_{\kw{simple}}}
\newcommand{\fastrewritingconstraints}{\Sigma_{\ptime}}
\newcommand{\ltgd}{\kw{LTGD}}
\newcommand{\gtgd}{\kw{GTGD}}
\newcommand{\fgtgd}{\kw{FGTGD}}

\newcommand{\uid}{\kw{UID}}
\newcommand{\incd}{\kw{IncDep}}
\newcommand{\atomicmap}{\kw{AtomMap}}
\newcommand{\projectionmap}{\kw{ProjMap}}

\newcommand{\guardedmap}{\kw{GuardedMap}}
\newcommand{\cqmap}{\kw{CQMap}}
\newcommand{\mappingClass}{\kw{Map}}
\newcommand{\policyClass}{\kw{Policy}}
\newcommand{\critinstrewrite}{\kw{CritRewrite}}
\newcommand{\critinstrewriteptime}{\critinstrewrite_{\ptime}}
\newcommand{\annot}{\kw{Annot}}

\newcommand{\hocwq}{\kw{HOCWQ}}

\newcommand{\facts}{\ensuremath{\mathcal{F}}}
\newcommand{\instance}{\ensuremath{\mathcal{D}}}

\newcommand{\schema}{\ensuremath{\mathcal{S}}}

\newcommand{\TuringMachine}{\ensuremath{\mathscr{T}}}

\newcommand{\letters}{\ensuremath{\Sigma}}
\newcommand{\states}{\ensuremath{Q}}
\newcommand{\initialState}{\ensuremath{q_0}}

\newcommand{\transitionFunction}{\ensuremath{\delta}}
\newcommand{\stateType}{\ensuremath{g}}

\newcommand{\Pspace}{\textsc{PSpace}}

\newcommand{\pspace}{\Pspace}
\newcommand{\PTime}{\textsc{PTime}}
\newcommand{\ptime}{\PTime}
\newcommand{\NP}{\textsc{NP}}
\newcommand{\np}{\NP}
\newcommand{\Exptime}{\textsc{ExpTime}}
\newcommand{\exptime}{\Exptime}
\newcommand{\expspace}{\textsc{ExpSpace}}
\newcommand{\hide}{\kw{Hide}}
\newcommand{\DoubleExptime}{\textsc{2ExpTime}}
\newcommand{\twoexp}{\DoubleExptime}
\newcommand{\twoexptime}{\twoexp}
\newcommand{\froneltgd}{\kw{Fr1}\ltgd}

\newcommand{\init}{\kw{Init}}

\newcommand{\quantif}{\ensuremath{\mathscr{Q}}}

\newcommand{\genAddr}{\kw{GenAddr}}
\newcommand{\data}{\kw{Cell}}
\newcommand{\children}{\kw{Children}}
\newcommand{\addr}{\kw{Address}}

\newcommand{\config}{\kw{Config}}
\newcommand{\transition}{\kw{Transition}}

\newcommand{\A}{\ensuremath{\mathcal{A}}\xspace}

\newcommand{\C}{\ensuremath{\mathcal{C}}\xspace}
\newcommand{\D}{\ensuremath{\mathcal{D}}\xspace}

\newcommand{\I}{\ensuremath{\mathcal{I}}\xspace}

\newcommand{\M}{\ensuremath{\mathcal{M}}\xspace}

\newcommand{\Q}{\ensuremath{\mathcal{Q}}\xspace}

\newcommand{\T}{\ensuremath{\mathcal{T}}\xspace}

\newcommand{\V}{\ensuremath{\mathcal{V}}\xspace}

\newcommand{\qentail}{\kw{QEntail}}
\newcommand{\kw}[1]{{\mathsf{#1}}\xspace}
\newcommand{\asecretquery}{p}
\newcommand{\secretquery}{\asecretquery} 

\newcommand{\disclose}{\kw{Disclose}}
\newcommand{\discloseClass}{\ensuremath{\kw{Disclose_C}}}

\newcommand{\critinst}{\instance_{\kw{Crit}}}

\newcommand{\critelement}{c_{\kw{Crit}}}

\newcommand{\iscrit}{\kw{IsCrit}}
\newcommand{\visiblechase}{\kw{VisChase}}
\newcommand{\chase}{\kw{Chase}}

\newcommand{\reach}{\kw{Reachable}}
\newcommand{\sourceschema}{{\cal S}}
\newcommand{\closedpreds}{{\cal C}}
\newcommand{\globalpreds}{{\cal G}}
\newcommand{\targetschema}{{\globalpreds(\M)}}
\newcommand{\viewcon}{\Sigma_{\M}}

\newcommand{\sourcecon}{\Sigma_{\kw{Source}}}

\newcommand{\org}{OR}

\newcommand{\isopen}{\kw{IsOpen}}
\newcommand{\docspec}{\kw{DocSpec}}
\newcommand{\patbldg}{\kw{PatBdlg}}
\newcommand{\patspec}{\kw{PatSpec}}
\newcommand{\patdoc}{\kw{PatDoc}}
\newcommand{\docbldg}{\kw{DocBldg}}
\newcommand{\openhours}{\kw{OpenHours}}
\newcommand{\visitinghours}{\kw{VisitingHours}}
\newcommand{\doctorlist}{\kw{DocList}}

\title{Reasoning about Disclosure in Data Integration in the Presence of Source Constraints}

\author{
  Michael Benedikt$^1$\and
  Pierre Bourhis$^2$\and
  Louis Jachiet$^{2}$\And
  Michaël Thomazo$^{3}$\\
  \affiliations
  $^1$University of Oxford\\
  $^2$CNRS CRIStAL, Université Lille, Inria Lille\\
  $^3$Inria, DI ENS, ENS, CNRS, PSL University\\
  \emails %
  \{pierre.bourhis, louis.jachiet\}@univ-lille.fr, %
  michael.thomazo@inria.fr, %
  michael.benedikt@cs.ox.ac.uk
}

\begin{document}

\maketitle

\begin{abstract}
Data integration systems allow users to 
access data sitting in multiple sources by means of queries over a
global schema, related to the sources via mappings.
Datasources often contain sensitive information, and thus an analysis
is needed to verify that a schema satisfies a privacy policy,
given as a set of queries whose answers should not be accessible to users.
Such an analysis should take into account not only knowledge that an
attacker may have about the mappings, but also what they may know
about the semantics of the sources.
In this paper, we show that \emph{source constraints} can have a
dramatic impact on disclosure analysis.
We study the problem of determining whether a given data integration
system discloses a source query to an attacker in the presence of
constraints, providing both lower and upper bounds on source-aware
disclosure analysis.
\end{abstract}

\section{Introduction}

In data integration,  users are shielded from the heterogeneity of
 multiple datasources by querying via
a \emph{global schema}, which provides a unified vocabulary. The relationship between sources and the user-facing schema
are specified declaratively via \emph{mapping rules}.
In data integration systems based on knowledge representation techniques, users pose queries
against the global schema, and these queries are answered using  data
in the sources and background knowledge. The computation of the answers involves reasoning based on the
query, the mappings, and any additional semantic information that is known  on the
global schema.

Data integration brings with it the danger of disclosing
information that data owners wish to keep confidential.
In declarative data integration, detection
of privacy violations is complex: although explicit access
to source information may be masked by the global schema, an attacker can
infer source facts via reasoning with schema and mapping information.

\begin{example} \label{ex:simple}
We consider an information integration setting for a hospital, which
internally stores the following data:

\begin{tabular}{r|l}
  Predicate & Meaning\\ \hline
 $\isopen(b,t)$    &    {building $b$ is open on date $t$} \\
  $\patbldg(p,b)$  &   {patient $p$ is present in building $b$} \\
 $\patspec(p,s)$   &   {patient $p$ was treated for specialty $s$} \\
 $\patdoc(p,d)$   &    {patient $p$ was treated by doctor $d$}  \\
 $\docbldg(d,b)$   &   {doctor $d$ is associated with building $b$} \\
 $\docspec(d,s)$  &    {doctor $d$ is associated with specialty $s$} \\
\end{tabular}\\

The hospital publishes the following data: $\openhours(b,t)$ giving
opening times $t$ for building $b$, $\visitinghours(p,t)$ giving
times $t$ when a given patient $p$ can be visited, and  $\doctorlist(d,s,b)$ 
listing the doctors $d$ with their specialty $s$ and their building $b$. Formally the
data being exposed is given by the following mappings:
\[ %
\begin{array}{rcl} %
  \isopen(b,t) & \rightarrow & \openhours(b,t)   \\
  \patbldg(p,b) \wedge \isopen(b,t) &\rightarrow & \visitinghours(p,t)\\
 \docspec(d,s) \wedge \docbldg(d,b) &\rightarrow& \doctorlist(d,s,b)  \\
\end{array}
\]
\end{example}

Prior work \cite{bckjournal} has studied disclosure in knowledge-based data integration, with an emphasis
on the role of semantic information on the \emph{global schema} -- in the form of ontological rules that
relate the global schema vocabulary.
The presence of an ontology can assist in privacy, since distinctions in
the source data may become indistinguishable
in the ontology.
More dangerous from the point of view of protecting information is \emph{semantic information about sources}.
For example, the sources in a data integration setting will generally overlap: that is, they will satisfy
\emph{referential integrity constraints}, saying that data items in one source link to items in another source.
Such constraints should be assumed as public knowledge, and with that knowledge the attacker may be able to infer 
information that was intended to be secret.

\begin{example} \label{ex:withconstraints}
Continuing Example \ref{ex:simple},
suppose that we know that each patient has a doctor specialized
in their condition, which can be formalized as:
$$\patdoc(p,d) \rightarrow \exists s ~ \patspec(p,s) \wedge \docspec(d,s)$$
And that we also know that when a patient is in a building, they must
have a doctor there:
$$\patbldg(p,b) \rightarrow \exists d ~\patdoc(p,d) \wedge \docbldg(d,b)$$

Due to the presence of these constraints, there can be a disclosure of the relationship of
patient to speciality $\patspec(p,s)$. Indeed, an attacker can see the
$\visitinghours$ for $p$, and from this, along with $\openhours$,  they  can sometimes infer the building $b$ where $p$
is treated (e.g. if $b$ has a unique set of open hours). From this they may be able to infer, using $\doctorlist$, the
specialty that $p$ has been treated for -- for example, if  all the doctors in $b$ share a
specialty.

\end{example}

In this work, we perform a detailed examination of the role of source constraints in disclosing
information in the context of data integration.
We focus on mappings from the sources given by universal Horn rules, where the global schema comes with no constraints.
Since our disclosure problem requires reasoning over all sources satisfying the constraints, we need a constraint
formalism that admits effective reasoning. We will look at  a variety of well-studied rule-based formalisms, with the simplest being
referential constraints, and the most complex being the \emph{frontier-guarded rules} \cite{frontierguarded}.
While decidability of our disclosure problems will follow from prior work \cite{bbcplics}, we will need new tools
to analyze the complexity of the problem. In Section \ref{sec:reduction},
we  give  reductions of disclosure problems to the \emph{query entailment problem}
that is heavily-studied  in knowledge representation. While a na\"ive application of the reduction allows us
only to conclude very pessimistic bounds, a more fine-grained analysis, combined with some recent results on CQ entailment,
will allow us to get much better bounds, in some cases ensuring tractability. 
In Section \ref{sec:lower}, we complement these results with lower bounds. 
Both the upper and lower bounds revolve around a complexity analysis
for reasoning with \emph{guarded existential rules and a restricted class of equality rules, where the rule
head compares a variable and a distinguished constant}. We believe this exploration of limited equality rules
can be productive for other reasoning problems.

Overall we get a complete picture  of the complexity
of disclosure in the presence of source constraints for many natural classes:  see Tables \ref{tab:summary}
in Section \ref{sec:conc} for a summary of our bounds.
Full proofs are available at the address \url{https://hal.inria.fr/hal-02145369}.


\section{Preliminaries} \label{sec:prelims}

We adopt standard notions from
function-free first-order
logic over a vocabulary of relational names.
An  \emph{instance} is a finite set of facts.
By a \emph{query} we always mean a \emph{conjunctive query} (CQ), which
is a first-order formula of the form $\exists \vec x ~ \bigwedge A_i$, where each $A_i$
is an atom.   
The \emph{arity} of a CQ is the number of its free variables, and CQs of arity 0 are \emph{Boolean}.

\paragraph{Data integration.}
Assume that the relational names in the vocabulary are split into two disjoint subsets:
\emph{source} and
\emph{global schema}. The \emph{arity} of such a schema is the maximal arity of its relational names.
We consider a set $\M$ of \emph{mapping rules} between source relations and a global schema relation $\T$ given. We focus on rules $\phi(\vec x, \vec y) \rightarrow \T(\vec x)$
where $\phi$ is a conjunctive query, there are no repeated variables in $\T(\vec x)$,
and where each global schema relation $\T$ is associated
with \emph{exactly one rule}.
Such rules are sometimes called ``GAV mappings''  in the database literature \cite{dataintegration}, and the
unique $\phi$ associated to a global relation $\T$ is referred to as the \emph{definition} of $\T$.
The rules are \emph{guarded} ($\M \in \guardedmap$) if for every rule, there exists an atom in the antecedent $\phi$ that contains
all the variables of $\phi$.
The rules are \emph{atomic} ($\M \in \atomicmap$) if each $\phi$ consists of a single atom, and 
they are \emph{projection maps} ($\M \in \projectionmap$) if each $\phi$ is a single atom with no repeated variables.
Given an instance $\D$ for the source relations, the \emph{image of $\D$ under mapping $\M$}, denoted $\M(\D)$, is
the instance for the global schema consisting of all facts $\{ \T(\vec c) \mid \D \models \exists \vec y ~ \phi(\vec c)\},$ where $\phi$ is the definition of $\T$.

\paragraph{Source constraints.}
We consider restrictions on the sources in the form of rules.
A \emph{tuple-generating dependency (TGD)} is a 
universally quantified sentence of the form
$\varphi(\mathbf{x}, \mathbf{z}) \rightarrow \exists \mathbf{y} \psi(\mathbf{x}, \mathbf{y})$, where
the \emph{body} $\varphi(\mathbf{x}, \mathbf{z})$ and the \emph{head} $\psi(\mathbf{x}, \mathbf{y})$  
are conjunctions of atoms such that each term is either a
constant or a variable in $\mathbf{x} \cup \mathbf{z}$ 
and $\mathbf{x} \cup \mathbf{y}$, respectively. Variables $\mathbf{x}$, common to the head and body, 
are called the \emph{frontier variables}. 
A  \emph{frontier-guarded TGD}  ($\fgtgd$) is a TGD in which there is an atom of the body that contains every frontier variable.
We focus on $\fgtgd$s because they
have been heavily studied in the database and knowledge representation community, and
it is known that many computational problems involving $\fgtgd$s are decidable \cite{frontierguarded}.
In particular this is true of the  \emph{query entailment problem}, which asks, given a finite collection of facts $\facts$,
a finite set $\Sigma$ of
sentences, and a CQ $Q$, whether $\facts \wedge \Sigma$ entails $Q$. We use $\qentail(\facts, \Sigma, Q)$ to denote
an instance of this problem and also say that ``$\facts$ entails $Q$ w.r.t. constraints $\Sigma$''.
A special case of $\fgtgd$s are \emph{Guarded TGDs} ($\gtgd$s), in which there is an atom containing all body variables.
These specialize further to \emph{linear TGDs}  ($\ltgd$s), whose   body consists of a single atom; and even further
to  \emph{inclusion dependencies} ($\incd$s), a linear TGD with a single atom in the head,
in which no variable occurs multiple times in the body, and no variable  occurs multiple times in  the head. 
Even $\incd$s occur quite commonly: for example, the source constraints of Example \ref{ex:withconstraints} can be rewritten
as $\incd$s.  The most specialized class we study are the \emph{unary $\incd$s}:($\uid$s), which are $\incd$s with at most one
frontier variable.

\paragraph{Queries and disclosure.}
The sensitive information in 
a data integration setting is given by a CQ $\secretquery$ over the
source schema, which we refer to as the \emph{policy}.
Intuitively, disclosure of sensitive information occurs
in a source instance $\D$ whenever the attacker can infer from the image $\M(\D)$
that $\secretquery$ holds of a tuple in $\D$.
Formally, we say an instance $\V$ for the global schema is \emph{realizable}, with respect to mappings $\M$ and source
constraints $\sourcecon$ if
there is some source instance $\D$ that satisfies $\sourcecon$ such that $\M(\D)=\V$. For a realizable $\V$, the set of
such $\D$ are the \emph{possible source instances for $\V$}.
A query result $\secretquery(\vec t)$ is \emph{disclosed} at $\V$ if $\secretquery(\vec t)$ holds on all possible source instances for $\V$.
A query $\secretquery$ \emph{admits a  disclosure} (for mappings $\M$ and source constraints $\sourcecon$) 
if there is some realizable instance $\V$ 
and binding $\vec t$ for the free variables of $\secretquery $ for which $\secretquery(\vec t)$ is disclosed.
In this terminology, the conclusion of Example  \ref{ex:withconstraints} was that  policy $\patspec(p,s)$ admits a disclosure
with respect to the constraints and mappings.
For a class of constraints $\C$, a class of mappings $\mappingClass$, a class of policies $\policyClass$, we write $\discloseClass(\C, \mappingClass)$ to denote the problem of determining whether a policy (a CQ, unless otherwise stated) admits a disclosure for a set of mappings in $\mappingClass$ and a set of source constraints in $\C$. Given $\sourcecon, \M$ and a CQ $\secretquery$, the corresponding instance of this problem is denoted by $\disclose(\C, \M, \secretquery)$.
 In this paper we will focus on disclosure for queries and constraints \emph{without constants}, although 
our techniques extend to the setting with  constants, as long as distinct constants
are not assumed to be unequal.


\section{Reducing Disclosure to Query Entailment} \label{sec:reduction}
Our first goal is to provide a reduction from $\discloseClass(\kw{TGD}, \mappingClass)$ to
a finite collection of standard query entailment problems.  
For simplicity
we will restrict to Boolean queries $\asecretquery$ in stating the results, but it is straightforward
to extend the reductions and results to the non-Boolean case.
We first recall a prior reduction of $\discloseClass(\kw{TGD}, \mappingClass)$ to a more
complex problem, the \emph{hybrid open and closed world query answering problem} 
\cite{hybrid1,hybrid2,hybrid3}, denoted $\hocwq$. $\hocwq$ takes as input a set of facts $\facts$, a collection of constraints $\Sigma$,
a Boolean query $Q$, and additionally a subset $\closedpreds$ of the vocabulary.  A \emph{possible world} for
such  $\hocwq(\facts, \Sigma, Q, \closedpreds)$  is any instance $\instance$ containing $\facts$, satisfying $\Sigma$, and such that
for each relation $C \in \closedpreds$, the $C$-facts in $\instance$ are the same as the $C$-facts in $\facts$.
$\hocwq(\facts, \Sigma, Q, \closedpreds)$ holds if $Q$ holds in every possible world. Note that
the query entailment problem is a special case of $\hocwq$, where $\closedpreds$ is empty.

Given a set of mapping rules $\M$ of the form $\phi(\vec y, \vec x) \rightarrow \T(\vec x)$, we let $\globalpreds(\M)$ be the set of global schema predicates, and
let $\viewcon(\M)$ be the mapping rules, considered as bi-directional constraints between global schema
predicates and sources.

We now recall one of the main results of \cite{bbcplics}:

\begin{theorem} \label{thm:critinst} There is an instance $\instance'$ 
 computable in linear time from $\sourcecon, \M, \asecretquery$, 
such that  $\disclose(\sourcecon, \M, \asecretquery)$ holds if and only if
$\hocwq(\instance', \sourcecon \cup \viewcon(\M), \asecretquery, \globalpreds(\M))$ holds.
\end{theorem}

In fact, the arguments in \cite{bbcplics} show that $\instance'$ can
be taken to be a very simple instance, the \emph{critical instance}
over the global schema $\globalpreds(\M)$ denoted
$\critinst^\targetschema$ where $\critinst^\schema$, for $\schema$ a
set of predicates, denotes the instance that mentions only a single
element $\critelement$, and contains, for each relation $R$ in
$\schema$ of arity $n$, the fact $R(\critelement, \ldots
\critelement)$.

\begin{corollary}  \label{cor:fgtgdupper} $\discloseClass( \fgtgd, \cqmap)$
is in $\twoexp$.
\end{corollary}
\begin{proof} 
The non-classical aspect of $\hocwq$ comes into play with rules of
$\viewcon(\M)$ of form $\phi(\vec x, \vec y) \rightarrow \T(\vec
x)$. But in the context of $\critinst^\targetschema$, these can be
rewritten as \emph{single-constant equality rules}
($\singleconstantrule$s) $\phi(\vec x, \vec y) \rightarrow \bigwedge_i
x_i=\critelement$. Such rules remain in the Guarded Negation Fragment
of first-order logic, which also subsumes $\fgtgd$s, while having a
query entailment problem in $\twoexp$ \cite{gnfj}.
\end{proof}

We now want to conduct a finer-grained analysis, looking for cases that give lower complexity.
 To do this we will transform further into a classical query entailment problem. This will require a transformation
of our query $\secretquery$, a transformation of our source constraints and mappings into a new set of constraints,
and a transformation of the instance $\critinst^\targetschema$. The idea of the transformation is that we
remove the $\singleconstantrule$s that are implicit in the $\hocwq$ problem, replacing them with constraints and queries that reflect all the possible
impacts the rules might have on identifying two variables.

We first describe the transformation of the query and the constraints.
They will involve introducing a new unary predicate $\iscrit(x)$; informally
this states that $x$ is equal to $\critelement$.
Consider a CQ $Q=\exists \vec y~ \bigwedge A_i$.
An \emph{annotation} of  $Q$
is a subset  of $Q$'s  variables.
Given an annotation $\annot$ of $Q$, we let $Q_\annot$ be the query
obtained from $Q$ by performing the following operation for
each $v$ in  $\annot$: for all occurrences $j$ of $v$ except the
first one, replacing $v$ with a fresh variable $v_j$; and adding  conjuncts $\iscrit(v_j)$ as well as $\iscrit(v)$ to $Q_\annot$.
A \emph{critical-instance rewriting} of a CQ $Q$ is a CQ obtained by applying the above process to $Q$ for any annotation.
We write $Q_\annot \in \critinstrewrite(Q)$ to indicate that $Q_\annot$ is such a rewriting.

To transform the mapping rules and constraints to a new set of constraints
using $\iscrit(x)$, we lift the notion of critical-instance rewriting to TGDs in the obvious way:
a  critical-instance rewriting of a TGD $\sigma$ (either in $\sourcecon$ or $\viewcon(\M)$),
is the set of  TGDs formed by applying the above process to the body of $\sigma$.
We write
$\sigma_\annot \in \critinstrewrite(\Sigma)$ to indicate that $\sigma_\annot$ is a critical-instance rewriting for a $\sigma \in \Sigma$,
and similarly for mappings. For example, the second mapping rule in Example \ref{ex:simple} has several rewritings;
 one of them will change the rule body 
to   $\patbldg(p,b) \wedge \isopen(b', d) \wedge \iscrit(b) \wedge \iscrit(b')$.

Our transformed constraints will additionally use the set of constraints $\iscrit(\M)$, including all rules:
\[
\T(x_1 \ldots x_n) \rightarrow \iscrit(x_i)
\]
where $\T$ ranges over the global schema and $1 \leq i \leq n$.
Informally $\iscrit(\M)$ states that all elements in the mapping image
must be $\critelement$.  We also need to transform the instance, using
a source instance with ``witnesses for the target facts''.  Consider a
fact $\T(\critelement \ldots \critelement)$ in
$\critinst^\targetschema$ formed by applying a mapping rule
$\bigwedge_i A_i(\vec x_i, \vec y_i) \rightarrow \T(\vec x)$ in $\M$.
The \emph{set of witness tuples for $\T(\vec x)$} is the set $A_i(\vec c)$, where $\vec c$ contains
$\critelement$ in each position containing a variable $x_j$ and containing a constant $c_{y_{j}}$ in every
position containing a variable $y_j$.
That is the witness tuples are witnesses for the fact $\T(\critelement \ldots \critelement)$, where each
existential witness is chosen  fresh.
Let $\hide_\M(\critinst^\targetschema)$ be the instance
formed by taking the witness tuples for every fact $\T(\critelement \ldots \critelement)\in\critinst^\targetschema$.

We are now ready to state the reduction of the disclosure problem to query entailment:

\begin{theorem} \label{thm:simplecritinstrewriting} 
$\disclose(\sourcecon, \M, \asecretquery)$ holds exactly when there is a $\secretquery_\annot \in \critinstrewrite(\secretquery)$ such that
$\hide_\M(\critinst^\targetschema)$ entails $\secretquery_\annot$ w.r.t. constraints:
\[
\critinstrewrite(\sourcecon) \cup \critinstrewrite(\M) \cup  \iscrit(\M)
\]
\end{theorem}

Note that  Theorem \ref{thm:simplecritinstrewriting} does not give a polynomial time reduction:
both $\critinstrewrite(\sourcecon)$ and $\critinstrewrite(\M)$ can contain exponentially many rewritings, and further
there can be exponentially many rewritings in $\critinstrewrite(\secretquery)$.

However, the algorithm does  give us a better bound in the case of Guarded TGDs with bounded arity.

\begin{corollary} \label{cor:gtgdatomicupper} 
If we bound the  arity of schema relations, then
$\discloseClass(\gtgd, \guardedmap)$ is in $\exptime$.
\end{corollary}
\begin{proof}
First, by introducing additional intermediate relations and source
constraints, we can assume that $\M$ contains only projection
mappings.  Thus we can guarantee that $\critinstrewrite(\M)$ just
contains the rules in $\M$.  By introducing intermediate relations and
additional source constraints, we can also assume that each $\gtgd \in
\sourcecon$ has a body with at most two atoms.
Since the arity of relations is fixed, the size of such $1$- or $2$-atom bodies is fixed as well.  From this we see
that the number of constraints in any $\critinstrewrite(\sigma)$ is polynomial. 
The reduction in Theorem \ref{thm:simplecritinstrewriting} thus gives us exponentially many $\gtgd$ entailment problems of polynomial size.
Since entailment over Guarded TGDs with bounded arity is in $\exptime$ \cite{tamingjournal}, we can conclude.
\end{proof}

\subsection{Refinements of the Reduction to Identify Lower Complexity Cases}
In order to lower the complexity to $\exptime$ \emph{without} bounding
the arity, we refine the construction of the function $\critinstrewrite(\sigma)$ in the case where $\sigma$ is a linear TGD, 
providing a function $\critinstrewriteptime(\sigma)$ that constructs only \emph{polynomially many rewritten constraints}.

Let $\sigma=B(\vec{x}) \rightarrow \exists \vec{y} ~ H(\vec{z})$ be
a linear TGD with relation $B$ of arity $k$,  and suppose $\vec{x}$ contains $d$ distinct free variables $V=\{v_1 \ldots v_d\}$.
Let $P$ be the set of pairs $(e,f)$ with  $e < f \leq k$ such that the same variable
$v_i$ sits at positions $e$ and $f$ in $\vec{x}$.
We order $P$ as $(e_0,f_0) \ldots (e_h,f_h)$; for each $(e,f)$ that is not the initial pair $(e_0, f_0)$, we
 let $(e,f)^-$ be its predecessor in the linear order.

 We let $B_{e,f}$ denote new predicates of arity $k$ for each $(e,f) \in P$.
Let $\vec w$ be a set of $k$ distinct variables, and $\vec w^{i=j}$ be formed from $\vec w$ by replacing $w_j$ with $w_i$.
We begin the construction of $\critinstrewriteptime(\sigma)$ with the constraints:
$B(\vec w^{e_0=f_0}) \rightarrow B_{e_0, f_0}(\vec w^{e_0=f_0})$ and
$B(\vec w) \wedge \iscrit(w_{e_0}) \wedge \iscrit(w_{f_0}) \rightarrow B_{e_0, f_0}(\vec w)$.

For each $(e,f)$ with a predecessor $(e,f)^-=(e', f')$, we add to $\critinstrewriteptime(\sigma)$ the following constraints:
$B_{e',f'}(\vec w^{e=f}) \rightarrow B_{e, f}(\vec w^{e=f})$ and
$B_{e',f'}(\vec w) \wedge \iscrit(w_{e}) \wedge \iscrit(w_{f}) \rightarrow B_{e, f}(\vec w)$.

Letting $e_h,f_h$  the final pair in $P$, we add to $\critinstrewriteptime(\sigma)$ the constraint
$ B_{e_h,f_h}(\vec x') \rightarrow \exists \vec{y} ~ H(\vec{z}) $
where $\vec x'$ is obtained from $\vec x$ by replacing all but the first
occurrence of each variable $v$  by a fresh variable.

If $\sourcecon$ consists of $\ltgd$s, we 
let $\critinstrewriteptime(\sourcecon)$ be the result of applying this process to every $\sigma \in \sourcecon$. Similarly,
if $\M$ consists of atomic mappings (implying that the associated  rules  are $\ltgd$s), then
we let $\critinstrewriteptime(\M)$ the result of applying the process above to the  rule going from
source relation to global schema relation associated to  $m \in \M$.
Then we have:

\begin{theorem} \label{thm:fastcritinstrewriting}
When $\sourcecon$ consists of $\ltgd$s,
$\disclose(\sourcecon, \M, \asecretquery)$ holds exactly when there is a $\secretquery_\annot \in \critinstrewrite(\secretquery)$ such that
$\hide_\M(\critinst^\targetschema)$ entails $\secretquery_\annot$ w.r.t. to the constraints
\[
\critinstrewriteptime(\sourcecon) \cup \critinstrewriteptime(\M) \cup  \iscrit(\M)
\]
\end{theorem}

We can combine this result with recent work on fine-grained complexity  of $\gtgd$s to improve the doubly exponential upper bound
of Corollary \ref{cor:fgtgdupper} for linear TGD source constraints and atomic mappings:

\begin{theorem} \label{thm:linlinupper} 
$\discloseClass(\ltgd, \atomicmap)$
is in $\exptime$. If the arity of relations in the source schema is bounded, then the complexity drops to $\np$, while
if further the policy is atomic, the problem is in $\ptime$.
\end{theorem}

\begin{proof}
It is sufficient to get an $\exptime$ algorithm for the entailment problem produced by Theorem \ref{thm:fastcritinstrewriting}, since
then we can apply it to each $\secretquery_\annot$ in $\exptime$.
The constraints in $\critinstrewriteptime(\sourcecon) \cup \critinstrewriteptime(\M)$ are Guarded TGDs that are not necessarily
$\ltgd$s. 
But the  bodies of these guarded
TGDs consist of a guard predicate and atoms over a  fixed ``side signature'', 
namely the unary predicate $\iscrit$.
It is known that the query entailment
for $\incd$s  and guarded TGDs with a fixed side signature
is in {\Exptime}, with the complexity dropping to $\np$ (resp. $\ptime$) when the arity is fixed (resp.
fixed and the query is atomic) \cite{antoinemichael}.
\end{proof}

Can we do better than $\exptime$? We  can note that
when the constraints $\sigma \in \sourcecon$ are $\incd$s,  $\critinstrewrite(\sigma)$ consists only of $\sigma$;
similarly if a mapping $m \in \M$ is a projection, then $\critinstrewrite(m)$ consists only of $m$.
This gives us a good upper bound in one of the most basic cases:

\begin{corollary} \label{cor:ididupper} 
$\discloseClass(\incd, \projectionmap)$
is in $\pspace$. If further a bound is fixed on the arity of relations in the source schema, then the problem becomes $\np$, dropping
to $\ptime$ when the policy is atomic.
\end{corollary}
\begin{proof}
Our algorithm will guess a $\secretquery_\annot$ in $\critinstrewrite(Q)$ and checks the entailment  of Theorem \ref{thm:simplecritinstrewriting}.
This gives  an  entailment problem for $\incd$s, known to be in $\pspace$ in general,  in $\np$ for bounded arity,
and in $\ptime$ for bounded arity and atomic queries \cite{johnsonklug}.
\end{proof}

\subsection{Obtaining Tractability}
Thus far we have seen cases where the complexity drops to $\pspace$ in the general case and $\np$ in the bounded
arity case, and $\ptime$ for atomic queries.   We now present a case where we obtain tractability for arbitrary
queries and arity. Recall that a $\uid$ is an $\incd$ where
at most  one variable  is exported. They are actually quite common, capturing
referential integrity when data is identified by a single attribute. We can show that restricting
to $\uid$s while having only projection maps leads to  tractability:
\begin{theorem} \label{thm:uidptime} 
$\discloseClass(\uid, \projectionmap)$
is in $\ptime$.
\end{theorem}
\begin{proof}
The first step is to  refine the reduction of Theorem~\ref{thm:simplecritinstrewriting} 
to get an entailment problem with only $\uid$s, over an instance consisting of a single
unary fact $\iscrit(\critelement)$. The main issue is  avoid the constraints 
in $\viewcon(\M)$, corresponding to the mapping rules. The intuition for this is that 
on 
$\critinst^\targetschema$, the  only impact of the backward and forward implications of $\viewcon(\M)$ 
is  to create new facts  among the source relations.  In
 these new facts only 
 $\critelement$, is propagated. 
Rather than creating $\singleconstantrule$s (implicitly what happens in the $\hocwq$ reduction) or generating classical constraints
where the impact of the equalities are ``baked in'' (as in the critical-instance rewritings of Theorems \ref{thm:simplecritinstrewriting}
and \ref{thm:fastcritinstrewriting}), we truncate
 the source relations to the positions where non-visible elements occur, while
generating $\uid$s on these truncated relations that simulate the impact of back-and-forth using $\viewcon(\M)$.

The second step is to show that query entailment with $\uid$s over the instance consisting only
of $\iscrit(\critelement)$ is in $\ptime$. This can be seen as an extension of the
$\ptime$  inference algorithm for $\uid$s \cite{uidptimeimp}. The idea behind this result is to analyze
the classical ``chase procedure'' for query entailment with TGDs \cite{faginetaldataex}.
In the case of $\uid$s over a unary fact, the shape of the chase model is very restricted; roughly speaking, 
it is a tree where only a single fact connects two values. Based on this, we can simplify the query
dramatically, making it into an acyclic query where any two variables co-occur in at most one predicate.
Once query simplification  is performed, we can reduce query entailment to polynomial many
entailment problems involving individual atoms in the query. This in turn can be solved using
the $\uid$ inference procedure of \cite{uidptimeimp}.
\end{proof}

\renewcommand*{\arraystretch}{1.1}

\begin{table*}[!h]
  \centering
\resizebox{\textwidth}{!}{ 
  \begin{tabular}{c|llll|llll}
\multicolumn{1}{c}{}    & \multicolumn{4}{c}{Unbounded arity} & \multicolumn{4}{c}{Bounded arity}  \\[-0.5em]
  \backslashbox{$\sourcecon$}{$\M$}   
  &      {$\projectionmap$}     & {$\atomicmap$} & {$\guardedmap$}   & {$\cqmap$}
  &      {$\projectionmap$}     & {$\atomicmap$} & {$\guardedmap$}   & {$\cqmap$}  \\[-0.3em] \cline{2-9}

  $\incd$    
  & {\Pspace}$_{L=\qentail}^{U={C\ref{cor:ididupper}}}$ 
  & \exptime$_{L=T\ref{thm:exptimelowergeneral}}$ 
  & \twoexptime$_{L=T\ref{thm:incdguardedlower}}$ 
  &  \twoexptime 
  & {\NP}$_{L=\qentail}$ 
  & \NP  
  & \exptime$_{L=T\ref{thm:incdguardedlower}}$  
  & \twoexptime$_{L=T\ref{thm:incdepcqmapboundedlower}}$
  \\
  $\ltgd$  
  & \exptime$_{L=T\ref{thm:exptimelowergeneral}}$  
  & \exptime$^{U=T\ref{thm:linlinupper}}$ 
  & \twoexptime 
  &  \twoexptime 
  & \NP  
  & {\NP}$^{U=T\ref{thm:linlinupper}}$  
  & \exptime 
  & \twoexptime
  \\
  $\gtgd$  
  & \twoexptime$_{L=T\ref{thm:incdguardedlower}}$ 
  & \twoexptime 
  & \twoexptime 
  &  \twoexptime 
  & \exptime$_{L=T\ref{thm:incdguardedlower}}$  
  & \exptime 
  &  \exptime$^{U=C\ref{cor:gtgdatomicupper}}$ 
  & \twoexptime
 \\
    $\fgtgd$ 
 & \twoexptime 
 & \twoexptime 
 & \twoexptime 
 & \twoexptime$^{U=C\ref{cor:fgtgdupper}}$
 & \twoexptime$_{L=\qentail}$
 & \twoexptime  
 & \twoexptime 
 &  \twoexptime$^{U=C\ref{cor:fgtgdupper}}$
 \\
\end{tabular}}
\caption{Complexity of disclosure:
  {\Pspace}$_{L=\qentail}^{U={C\ref{cor:ididupper}}}$ means the corresponding
  problem is $\Pspace$-complete, where the {U}pper bound is given by
  {C}orollary \ref{cor:ididupper} (U=C\ref{cor:ididupper}) and 
  the {L}ower bound is inherited from  entailment. We omit bounds
  inferred from inclusion ($\M$ or $\sourcecon$).}
\label{tab:summary}
\end{table*}

 \section{Lower Bounds} \label{sec:lower}
We now focus on providing lower bounds for $\discloseClass(\C, \mappingClass)$, showing in particular that the upper bounds
provided in Section \ref{sec:reduction} can not be substantially improved. For many
classes of constraints it is easy to see that  the complexity of
disclosure inherits the lower bounds for the classical entailment problem
for the class. From this we get a number of matching lower bounds; e.g. $\twoexp$ for $\gtgd$ constraints,
$\pspace$ for $\incd$ constraints.
But note that in some cases the  upper bounds we have provided for disclosure
in Section \ref{sec:reduction} are  higher than the complexity of entailment over the source constraints. For example,
for $\incd$s we have provided only a $\twoexp$  upper bound for guarded mappings (from Corollary
\ref{cor:fgtgdupper}),  and only an exponential bound for atomic mappings (from Theorem \ref{thm:linlinupper}).
This suggests that the form of the mappings influences the complexity as well, as we now show.

Most of our proofs for hardness above the entailment bound for source constraints
rely on the encoding of a Turing machine.  Source constraints are 
used to generate the underlying structures (tree of configurations, 
tape of a Turing machine) while mappings are used to ensure consistency (a universal configuration is accepting if and only if all its successor configurations are accepting, the content of the tape is consistently represented,...). To illustrate our approach, we sketch the proof of the following result.

\begin{theorem}
  \label{thm:incdguardedlower}
 $\discloseClass(\incd,\guardedmap)$
and $\discloseClass(\gtgd,\projectionmap)$
 are $\twoexp$-hard,
 and are $\exptime$-hard even in bounded arity.
\end{theorem}
\begin{proof}
Recall that Theorem \ref{thm:critinst} relates disclosure to a $\hocwq$ problem on $\critinst^\targetschema$. Also recall from Section \ref{sec:reduction} the intuition that such a problem amounts to a classical entailment problem 
for a CQ  over a very simple
instance,  using the source dependencies and $\singleconstantrule$s: of the
form  $\phi(\vec{x}) \rightarrow x=\critelement$, where $\phi$  will be the body of a mapping. 
We will sketch how
 to simulate an alternating $\expspace$ Turing machine $\M$ using a $\qentail$ problem
using $\incd$s  and \emph{guarded} $\singleconstantrule$s. This can in turn be simulated using our $\hocwq$ problem.

We first build a tree of configurations using $\incd$s, such that each node has a type (existential or universal) and is the parent of two nodes (called $\alpha$-successor and $\beta$-successor) of the opposite type. This tree structure is represented, together with additional information, by atoms such as:
\[\children_\forall(c,c_\alpha,c_\beta,ac,ac_\alpha,ac_\beta,\vec{y_0},\vec{y_1},r).\]
Intuitively, this states that $c$ is a universal configuration, parent of $c_\alpha$ and $c_\beta$. $ac$ (resp. $ac_\alpha$, resp. $ac_\beta$) is the acceptance bit for $c$ (resp. $c_\alpha$, resp. $c_\beta$), which will be made equal to $\critelement$ if and only if the configuration represented by $c$ (resp. $c_\alpha$, resp. $c_\beta$) is accepting. $\vec{y_0},\vec{y_1}$ will be used to represent cell addresses, while $r$ is the identifier of the root of the configuration tree. The initial instance is such an atom, where the first position and the last position are the same constant, $\vec{y_0}$ is a vector of $n$ $0$'s, $\vec{y_1}$ is a vector of $n$ $1$'s, and all other arguments are distinct constants. 

We use $\singleconstantrule$s  to propagate acceptance information up in the tree. For instance, a universal configuration is accepting if both its successors are accepting. This is simulated by the following $\singleconstantrule$:
\[\children_\forall(c,c_\alpha,c_\beta,ac_c,\critelement,\critelement,\vec{y_0},\vec{y_1},r) \rightarrow ac_c = \critelement.\]

To simulate $\M$, we need access to an exponential number of cells for each configuration. We identify a cell by the configuration it belongs to and an address, which is a vector, generated by $\incd$s, of length $n$ whose arguments are either $0$ or $1$. 
The atom for representing a cell is thus
$\data(c_p,c,\vec{addr},\vec{v},\vec{v}_{prev},\vec{v}_{next})$,
where $c_p$ is the parent configuration of $c$, which is the configuration to which the represented cell 
belongs, $\vec{addr}$ is the address of the cell, $\vec{v}$ its content, $\vec{v}_{prev}$ the content of the previous 
cell, and $\vec{v}_{next}$ the content of the next cell. Note that this representation is redundant, and we need to use 
$\singleconstantrule$s  to ensure its consistency.

Note that $\vec{v}$ is a tuple of length the size of $(\letters\cup\{\flat\})\times(\states\cup\{\bot\})$. Each position corresponds to an element of that set, and the content of a represented cell is the element which corresponds to the unique position in which $\critelement$ appears. 

We now explain how to build the representation of the initial tape, and simulate the transition function. 
Both steps are done by unifying some nulls with $\critelement$. W.l.o.g., we assume that the initial tape 
contains a $l$ in the first cell, on which points the head of $\M$ in a state $s$, and that $(l,s)$ corresponds to the first bit of $\vec{v}$. We thus use a $\singleconstantrule$ to set this bit to 
$\critelement$ in the first cell of the first configuration. We then set (w.l.o.g.) the second bit of all the other cells of that 
configuration to $\critelement$ (assuming this represents $(\flat,\bot)$).

To simulate the transitions, we note that the content of a cell in a configuration depends only on the content of 
the same cell in the parent configuration, along with the content of parent's previous and next cells.
 We thus add a $\singleconstantrule$ that checks
 for the presence of $\critelement$ specifying the content of three consecutive cells in a configuration, and unify a null with $\critelement$ to specify the content of the corresponding cell of a child configuration.

The argument above uses $\incd$s and $\guardedmap$s, but we can simplify the mappings to $\projectionmap$ using
$\gtgd$s. \end{proof}

A simple variation of the construction used for $\pspace$-hardness of entailment
with $\incd$s \cite{casanova} shows that our upper bounds for $\incd$ source constraints and atomic maps
are tight. The case of $\ltgd$ source constraints and projection maps can be done via reduction
to that of  $\incd$ source constraints and atomic maps:

\begin{theorem} \label{thm:exptimelowergeneral}
$\discloseClass(\incd,\atomicmap)$ 
and $\discloseClass(\ltgd, \projectionmap)$
are both $\exptime$-hard. 
\end{theorem}

The above results, coupled with  argument that the lower bounds for entailment are inherited by disclosure, show tightness of all upper bounds from Table \ref{tab:summary} in the unbounded arity case.
Another variation  of the encoding in Theorem \ref{thm:incdguardedlower}
 shows that with no restriction on the mappings 
one can not do better than the  $\twoexp$ upper bound
of Corollary \ref{cor:fgtgdupper} even for $\incd$ constraints in bounded arity, 
\begin{theorem} \label{thm:incdepcqmapboundedlower}
  $\discloseClass(\incd,\cqmap)$ 
   is $\twoexp$-hard  in bounded arity.
\end{theorem}

%
The  theorem above, again combined with results showing that the lower bounds for
entailment are inherited, suffice
to show tightness of all upper bounds from Table \ref{tab:summary}
in the case of bounded arity.

We can also show that our tractability result for $\uid$ constraints and projection maps
does not extend when either the maps or the constraints are broadened.
Informally, this is because with these extensions we can generate an instance
on which CQ querying  is $\np$-hard.



\section{Related Work} \label{sec:related}

Disclosure analysis has been approached from many angles.
We do not compare with the vast amount of work
that  analyzes probabilistic mechanisms for releasing information,  providing probabilistic guarantees on disclosure \cite{DBLP:conf/icalp/Dwork06}.
Our work focuses on the impact of reasoning on mapping-based mechanisms used in knowledge-based information integration, which are
deterministic; thus one would prefer, and can hope for, deterministic guarantees on disclosure.
We deal here with the \emph{analysis} of disclosure, while there is a complementary literature on  how to \emph{enforce} privacy
\cite{DBLP:journals/ijisec/BiskupW08,DBLP:journals/tkde/BonattiKS95,DBLP:conf/semweb/BonattiS13,DBLP:journals/tdp/StuderW14}.
   
The problem of whether information is disclosed on a particular instance (variation of $\hocwq$ introduced in Section \ref{sec:reduction})
has been studied in both the knowledge representation \cite{hybrid1,hybrid2,hybrid3,ahmetajclosed,amendolaclosed}
and database community
\cite{abiteboulduschka}.
The corresponding schema-level problem was defined in  \cite{bbcplics}, which
allows arbitrary constraints  relating the source and the global schema. However,   results are provided only for constraints in guarded logics, which
does not subsume the case of mappings given here. 
Our results clarify some issues in  prior work: \cite{bbcplics} claimed that
disclosure with $\incd$ source constraints and atomic maps is in $\pspace$, while our Theorem
\ref{thm:exptimelowergeneral}
 shows that the problem is $\exptime$-hard.
Our notion  of disclosure 
corresponds to the complement of \cite{bckjournal}'s ``data-independent compliance''. The formal
framework of \cite{bckjournal}  is orthogonal to ours. On the one hand, source constraints are absent; on the other hand a more powerful mapping
language is considered, with existentials in the head of rules, while
constraints on the global schema, given by ontological axioms, are now allowed.  
\cite{bckjournal} assume
that the attacker has an interface for posing queries against the global schema, with the queries being answered under entailment semantics. 
In general, the semantic information on the global schema makes disclosure harder, since the outputs of different mapping
rules may be indistinguishable by an attacker who only sees the results of reasoning.
In contrast, source constraints make disclosure of secrets easier, since they provide additional information to the attacker.

\section{Summary and Conclusion} \label{sec:conc}
We have isolated the complexity of information disclosure from a schema in the presence
of commonly-studied sets of source constraints. A summary of many
combinations of mappings $\M$ and source constraints $\sourcecon$ is given
in Table~\ref{tab:summary}: note that \emph{all problems are complete for the complexity classes
listed}. We have shown tractability in the case of $\uid$s and projection maps (omitted in
the tables), while showing that lifting the restriction leads to intractability. But we leave open  a
finer-grained analysis of complexity for frontier-one constraints with more general mappings.
Our results depend on a fine-grained analysis of reasoning with TGDs and $\singleconstantrule$s, a topic
we think is of independent interest.

\section* {Acknowledgements}

This work was partially funder by CNRS Momentum project ``Managing Data without Leak''.

\bibliographystyle{named}
\bibliography{bib}
\pagebreak
\appendix
\section{Detailed Proofs from Section \ref{sec:reduction}: Upper Bounds for Disclosure}
\subsection{Proof of Theorem \protect{\ref{thm:simplecritinstrewriting}: Correctness of the Basic Reduction
from Disclosure to Classical Entailment}
}

Recall the statement of Theorem \ref{thm:simplecritinstrewriting}, which
applies the algorithms $\critinstrewrite(\sourcecon)$ to TGDs
and $\critinstrewrite(\M)$ to mappings.

\medskip

$\disclose(\sourcecon, \M, \asecretquery)$ holds exactly when there is a $\secretquery_\annot \in \critinstrewrite(\secretquery)$ such that
$\hide_\M(\critinst^\targetschema)$ entails $\secretquery_\annot$ w.r.t. constraints:
\[
\critinstrewrite(\sourcecon) \cup \critinstrewrite(\M) \cup  \iscrit(\M)
\]
holds.

\medskip

By Theorem \ref{thm:critinst} we know that $\disclose(\sourcecon, \M, \asecretquery)$
is equivalent to  $\hocwq(\critinst^\targetschema, \sourcecon \cup \viewcon(\M), \asecretquery, \globalpreds(\M))$.

This will immediately allow us to prove one direction of the equivalence.
Suppose each of our entailments fails. From this, we see using \cite{sagivy} that
$\hide_\M(\critinst^\targetschema)$ does not entail the disjunction of $\secretquery_\annot$. 
Thus we have an instance $\instance$ extending $\hide_\M(\critinst^\targetschema)$  with facts
that may include the $\iscrit$ predicate, where $\instance$  satisfies all the rewritten constraints and no rewritten query
$\secretquery_\annot$. Note that since $\instance$ satisfies  the constraints of $\critinstrewrite(\M)$ as well as 
$\iscrit(\M)$, we know that the element $\critelement$, if it occurs in $\hide_\M(\critinst^\targetschema)$,  will be labeled with $\iscrit$.

Form an instance $\instance'$  by unifying all elements $e$ in $\instance$ satisfying $\iscrit$ into a single
element $\critelement$, making $\critelement$ inherit any fact that such an $e$ participates in. That is, we choose
$\instance'$ so that if $h$ is the mapping taking any element satisfying $\iscrit$ to $\critelement$ and fixing every
other element, then $h$ is a homomorphism from $\instance$ onto $\instance'$.
We can easily verify that $\instance'$  satisfies the original source constraints $\sourcecon$. For  each homomorphism
$\lambda'$ of the body of $\sigma' \in \sourcecon$ into  $\instance'$, there is a homomorphism $\lambda$ of some
$\sigma \in \critinstrewrite(\sigma')$ into $\instance$. We know $\sigma$ is satisfied in $\instance$, and taking
the $h$-image of the tuples that witness this gives us the required witnesses for $\sigma'$ in $\instance'$.
Now let $\instance'_0$ be the  restriction of $\instance'$ to the source relations. We argue that
the mapping image of $\instance'_0$  under $\M$ is exactly $\critinst^\targetschema$.
To see that the image of $\instance'_0$ must include all the facts  in $\critinst^\targetschema$, 
note
that $\instance$ includes all facts of $\hide_\M(\critinst^\targetschema)$,
which contains witnesses for each such fact. Thus the $h$-image, namely $\instance'$, contains witnesses for each
such fact as well.
Conversely, suppose the image of $\instance'_0$ includes a  fact $F(\vec d)$; we will argue that $F(\vec d)$ is
in $\critinst^\targetschema$.
 Since $\instance$ satisfied  $\iscrit(\M)$,  any such fact in $\instance$ must have all $d_i$ satisfying $\iscrit$.
Thus in $\instance'_0$ each such fact
 must be of the form $F(\critelement \ldots \critelement)$. Thus the
$\M$-image of $\instance'_0$ is exactly the same  $\critinst^\targetschema$.

Finally, we claim that $\instance'$ satisfies $\neg \asecretquery$. If it satisfies $\asecretquery$, then
$\instance$ would satisfy $\asecretquery_\annot$ for some annotation $\annot$, a contradiction.
Putting this all together, we see that $\instance'$ contradicts  $\hocwq(\critinst, \sourcecon \cup \viewcon(\M), \asecretquery, \globalpreds(\M))$.

Before turning to the other direction, we will explain some other results that will be necessary.
The first is the \emph{chase procedure} for checking  entailment of a query $Q$ from
a set of constraints $\Sigma$ and a set of facts $\instance$.  This proceeds by building a sequence of instances
$\instance= \instance_0 \ldots \instance_i \ldots$ where each $\instance_{i+1}$ is formed from $\instance_i$ by
``firing a rule'' $\sigma \in \Sigma$ $\instance_i$.  Firing $\sigma$ in $\instance_i$
means finding a homomorphism $\lambda$ from the body of $\sigma$ into $\instance_i$, and adding facts to extend $\lambda$
to the head, using fresh values for all existentially quantified variables. Such a homomorphism $\lambda$
is called a \emph{trigger} for the rule firing. The chase of
$\instance$ under $\Sigma$, denoted $\chase_\Sigma(\instance)$,  is any instance formed as the union of such a sequence having the additional property
that every rule that could fire in some $\instance_i$ fires in some later $\instance_j$.
The significance of the chase
for query entailment is the following result \cite{faginetaldataex}:
\begin{theorem} \label{thm:chase}
For an instance $\instance$, set of TGDs $\Sigma$, and UCQ $Q$, 
we have $\qentail(\instance, \Sigma, Q)$ if and only if some chase model for $\instance$ under $\Sigma$
satisfies $Q$.
\end{theorem}

We will also need  a variation of the chase for the problem
  $\hocwq(\critinst^\targetschema, \sourcecon \cup \viewcon(\M), \asecretquery, \globalpreds(\M))$, 
taken from \cite{bbcplics}. The \emph{visible chase} is a sequence of source instances
$\instance_0 \ldots \instance_n \ldots$ that begins with $\instance_0 =\hide_\M(\critinst^\targetschema)$.
$\instance_{i+1}$ is formed from $\instance_i$ by ``chasing and merging''.
The chase step applies the usual chase procedure
described above to  $\instance_i$ with constraints $\sourcecon$, creating new facts that possibly contain fresh values.
In a merge step, we take a mapping $m \in \M$ and  a homomorphism $\lambda$ of the body
of $m$ into $\instance_i$, and for each free variable $x$ of $m$, we 
replace $\lambda(x)$ by $\critelement$ in all facts in which it appears.  We say that this is a 
\emph{merge step with $m, \lambda$ on $\instance_i$}.
Since the process is monotone, it must
reach a fixpoint, which we refer to as the \emph{visible chase} of $\critinst^\targetschema$, denoted
 $\visiblechase(\sourcecon, \M)$.
\begin{proposition} \label{prop:visiblechase} \cite{bbcplics}
 $\hocwq(\critinst^\targetschema, \sourcecon \cup \viewcon(\M), \asecretquery, \globalpreds(\M))$ holds
exactly when $\visiblechase(\sourcecon, \M)$ satisfies $\asecretquery$.
\end{proposition}

We now prove the other direction, assuming that $\hocwq(\critinst^\targetschema, \sourcecon \cup \viewcon(\M), \asecretquery, \globalpreds(\M))$ fails,
but one of the entailments holds. By Theorem \ref{thm:chase}, this means that
some chase of $\hide_\M(\critinst^\targetschema)$ under the constraints $\critinstrewrite(\sourcecon) \cup \critinstrewrite(\M) \cup  \iscrit(\M)$
satisfies $\secretquery_\annot$ for some annotation $\annot$.
Let $\instance'_0 \ldots \instance'_n \ldots $ denote such a chase sequence for
$\hide_\M(\critinst^\targetschema)$ under $\critinstrewrite(\sourcecon) \cup \critinstrewrite(\M) \cup  \iscrit(\M)$.
We form another sequence $\instance_0 \ldots \instance_n \ldots$, with $\instance_0=\instance'_0$,
maintaining the invariant that there is a homomorphism $h_i$ from $\instance'_i$ to $\instance_i$ mapping every element satisfying 
$\iscrit$ to $\critinst^\targetschema$. The inductive step is performed as follows:
\begin{itemize}
\item For every  chase step with a rule $\sigma'$ of $\critinstrewrite(\sourcecon)$ applied in $\instance'_i$, having trigger
$\lambda'$, we know that $\sigma'=\critinstrewrite(\sigma)$ for some $\sigma \in \sourcecon$.
We can apply 
the  corresponding rule $\sigma$ in $\instance_i$, with a  trigger $\lambda$ that maps a variable $x$
to the $h_i$-image of $\lambda'(x)$. 
Thus  $\lambda$ composed with $h_i$ is $\lambda$.
\item For every chase step in $\instance'_i$ with a rule of $\sigma' \in \critinstrewrite(m)$ for $m \in \M$ and a trigger $\lambda$, 
we apply a merge step in $\instance_i$ with $m$ and $\lambda$.
\end{itemize}
Since some $\instance'_n$  satisfies $\secretquery_\annot$, one of the $\instance_n$ must satisfy $\secretquery_\annot$.
Since $\instance_n$ contains the image of $\instance'_n$ under the homomorphism $h_n$, and $h_n$ maps $\secretquery_\annot$
to $\secretquery$, we see that $\instance_n$ must satisfy $\secretquery$.
But $\instance_n$ is a subinstance of the visible chase for our $\hocwq$ problem. Thus the assumption that
 $\hocwq(\critinst^\targetschema, \sourcecon \cup \viewcon(\M), \asecretquery, \globalpreds(\M))$ fails and
Proposition \ref{prop:visiblechase} imply that $\secretquery$ cannot hold in $\instance_n$, a contradiction.

\subsection{Simplifying Mappings}
\label{ss:simplifymappings}
In this section, we will see that we can simplify mapping to
be projection maps at the cost of moving to  a richer
 class of source constraints.

Given a problem $\disclose(\sourcecon, \M, \asecretquery)$ we consider
$\sourcecon'$ and $\M'$ built in the following way: $\sourcecon'$ is
composed of $\sourcecon$ plus for each mapping $\phi(\vec x, \vec
y)\rightarrow T(\vec x)$ we create a predicate $R_\phi(\vec x,\vec y)$
and we add to $\sourcecon'$ the two constraints $\phi(\vec x, \vec
y)\rightarrow R_\phi(\vec x,\vec y)$ and $R_\phi(\vec x,\ \vec y)
\rightarrow \phi(\vec x, \vec y) $. $\M'$ is composed of mappings
$R_\phi(\vec x,\vec y)\rightarrow T_{\phi}(\vec x)$.

\begin{proposition}
\label{prop:reduceprojmap}
  We have $\disclose(\sourcecon, \M, \asecretquery)$ if and only if
$\disclose(\sourcecon', \M', \asecretquery)$.
\end{proposition}

\begin{proof}

  To prove the proposition, it is sufficient to prove that
 $\asecretquery$ holds on $\visiblechase(\sourcecon, \M)$ if and only
 if $\asecretquery$ holds on $\visiblechase(\sourcecon', \M')$ (see
 Proposition~\ref{prop:visiblechase}).  Let $\Pi(\D)$ be the instance
 obtained by removing all the facts $R_{\phi}(\vec x, \vec y)$ in
 $\D$.

  We recall that the visible chase works iteratively, at each step a
  database $\D_{i+1}$ is created from $\D_i$ by chasing all facts then
  merging some values with $\critelement$. For the sake of simplicity
  we suppose that each step is composed of either one rule firing or
  one merging.

  \begin{itemize}
  \item We start by proving that $\disclose(\sourcecon, \M,
    \asecretquery)$ implies $\disclose(\sourcecon', \M',
    \asecretquery)$.
  
  Let $\D_0,\dots$ be a sequence corresponding to
  $\visiblechase(\sourcecon,\M)$. We build a sequence $\D'_0,\dots$
  corresponding to $\visiblechase(\sourcecon',\M')$. We are trying to
  build $\D'_0, \dots $ such that there exists for all $i$ there
  exists $j$ such that $\D_i = \Pi(\D'_{j})$, and $h(x)=\critelement$
  implies $x=\critelement$.

We prove by induction:
\begin{itemize}
\item 
$\D_0$ is composed of witnesses of $\M$ and $\D'_0$ of witnesses of
  $\M'$. We build $\D'_1, \dots, \D'_j$ such that each $\D'_i$ is
  obtained by firing the $i$-th rule $R_\phi(\vec x,\vec y)\rightarrow
  \phi(\vec x,\vec y)$.

\item 
Let us suppose that $\D_i= \Pi(\D'_j)$ and $\D_{i+1}$ is
obtained by firing a rule $\sigma$; $\sigma$ could have been fired on
$\D'_j$ and thus we can build $\D'_{j+1}$  such that
$\D_{i+1}=\Pi(\D'_{j+1})$.
\item When $\D_{i+1}$ is obtained by merging values then it means that
  we have $\phi(\vec x,\vec y)$ holding in $\D_i$ and thus $\phi(\vec
  x,\vec y)$ holding in $\Pi(D'_j)$ therefore we could use the rule
  $\phi(\vec x,\vec y)\rightarrow R_\phi(\vec x,\vec y)$ followed by
  an unification on $R_\phi$. Therefore we can build
  $\D'_{j+1}=\D'_j\cup\{R_\phi(\vec x,\vec y)\}$ and $\D'_{j+2}$ such
  that $\D_{i+1}=\Pi(\D_{j+2})$.
\end{itemize}

\item For the direction $\disclose(\sourcecon', \M', \asecretquery)$
  implies $\disclose(\sourcecon, \M, \asecretquery)$ we start by
  noticing that, without loss of generality, we can suppose that the
  sequence $\D'_0,\dots$ of $\visiblechase(\sourcecon',\M')$ starts by
  firing each rule $R_\phi(\vec x,\vec y)\rightarrow \phi(\vec x,\vec
  y)$ (it is always possible to generate more facts) and then we
  create $\D_0, \dots$ such that for all $i$ big enough there exists
  $j$ such that $\D_j=h(\Pi(\D'_i))$

  \begin{itemize}
  \item Once all rules $R_\phi(\vec x,\vec y)\rightarrow \phi(\vec
    x,\vec y)$ have been fired, we see that we obtain an instance
    isomorphic to $\D_0$. 
    \item When $\D_j=h(\Pi(\D'_i))$ and $\D'_{i+1}$ is obtained
      through a merge step, it means that we had $\D'_i \models
      R_{\phi}(\vec x, \vec y)$ but we easily see by induction that
      this means that we had $\D_j \models h(\phi(\vec x, \vec y))$
      and thus that we can also perform the merge step on $\D_j$
    \item When $\D'_{i+1}$ is obtained through a rule, it is either a
      rule in $\sourcecon$ that we can reproduce in $\D_j$ or it is a
      rule $\phi(\vec x,\vec y)\rightarrow R_\phi(\vec x,\vec y)$. In
      this latter case, we don't have anything to do as $R_\phi(\vec
      x,\vec y)$ will be discarded by $\Pi$.
  \end{itemize}
  Now, we also see that $j$ will grow as $i$ grows since except for
  rules $\phi(\vec x,\vec y)\rightarrow R_\phi(\vec x,\vec y)$, our
  $j$ increases. Therefore at the limit we have that
  $\visiblechase(\sourcecon',\M')\models \asecretquery$ implies
  $\visiblechase(\sourcecon,\M)\models \asecretquery$.
  \end{itemize}

\end{proof}

\begin{corollary}\label{cor:guardedreduce}
$\discloseClass(\gtgd, \guardedmap)$ reduces to
  $\discloseClass(\gtgd, \projectionmap)$.
\end{corollary}

\subsection{More Details for the Proof of Corollary~\ref{cor:gtgdatomicupper}}

We recall the statement of Corollary~\ref{cor:gtgdatomicupper}: 

\medskip

If we
fix the maximal arity of relations in the schema, then
$\discloseClass(\gtgd, \guardedmap)$ is in $\exptime$.

\medskip

We now fill in the details of the proof sketch in the body.

\paragraph{Reducing to $\projectionmap$.}
Using Corollary~\ref{cor:guardedreduce}, we can reduce the problem to
$\discloseClass(\gtgd, \projectionmap)$. We now  show
that this latter problem is in $\exptime$.

\paragraph{Reducing to two atoms in the body of TGDs.}
Given a set of $\gtgd$s  $\sourcecon$ and a set of maps
$\M\in\incd$ we now reduce $\disclose(\sourcecon,\M,\asecretquery)$ to
$\disclose(\sourcecon',\M,\asecretquery)$ where each $\gtgd$ in
$\sourcecon'$ holds at most two conjuncts in the rule body.

$\sourcecon'$ is composed by applying the following process for each
$\gtgd$ $\phi(\vec x)\rightarrow \exists \vec y ~ R(\vec t)\in
\sourcecon$. The constraint $\phi(\vec x)\rightarrow \exists \vec y ~
R(\vec t)$ is guarded, therefore we can select a guarding conjunct
$G_\phi(\vec x)$ such that $\phi(\vec x) = G_\phi(\vec x) \land
Q_1(\vec x) \land \dots \land Q_k(\vec x)$. When $k\leq 1$ we simply
add $\phi(\vec x)\rightarrow \exists \vec y ~ R(\vec t)$ to
$\sourcecon'$. When $k>1$, we rewrite this constraint by introducing
$k$ predicates $R_1, \dots, R_{k}$, while producing the following
constraints $G_\phi(\vec x) \land Q_1(\vec x) \rightarrow R_1(\vec x)$
and for $1\leq i\leq k-1$: $R_i(\vec x) \land Q_{i+1}(\vec x)
\rightarrow R_{i+1}(\vec x)$. Finally we also add $R_k(\vec
x)\rightarrow \exists \vec y ~ R(\vec t)$. It is easy to see that this
new problem is equivalent because each constraint in $\sourcecon$ is
implied by its corresponding constraints in $\sourcecon'$ and if we
look at the result of the visible chase, the only fact derived from a
$R_k(\vec x)$ are facts $R(\vec y)$ such that $\phi(\vec x)$.

\paragraph{Rewriting in $\ptime$.}

Now that maps are $\projectionmap$s and each $\gtgd$ has at most two
atoms in their body, we can apply the rewriting presented in
Theorem~\ref{thm:simplecritinstrewriting}. Notice that each $\gtgd$ will
be rewritten to a bounded number of $\gtgd$s,  and the rewriting of 
the maps will be trivial. Since query entailment with $\gtgd$s is
$\Exptime$ when the arity is bounded we can conclude the proof.

\subsection{Proof of Theorem \protect{\ref{thm:fastcritinstrewriting}}: More Efficient Reduction to Entailment for $\ltgd$ Source Constraints and Atomic Mappings}
Recall the statement of Theorem \ref{thm:fastcritinstrewriting}, which concerns
the application of the rewriting algorithms $\critinstrewriteptime(\sourcecon)$ for $\ltgd$ source constraints
$\sourcecon$, 
and the algorithm $\critinstrewriteptime(\M)$ for atomic mappings $\M$:

\medskip

$\disclose(\sourcecon, \M, \asecretquery)$ holds exactly when there is a $Q_\annot \in \critinstrewrite(Q)$ such that
$\hide_\M(\critinst^\targetschema)$ entails $Q_\annot$ w.r.t. to the constraints
\[
\critinstrewriteptime(\sourcecon) \cup \critinstrewriteptime(\M) \cup  \iscrit(\M)
\]

\medskip

Let $\simplerewritingconstraints=\critinstrewriteptime(\sourcecon) \cup \critinstrewriteptime(\M) \cup  \iscrit(\M)$
and $\fastrewritingconstraints$ be the constraints posed in Theorem \ref{thm:fastcritinstrewriting}.
By Theorem  \ref{thm:simplecritinstrewriting}, it is enough to show that query entailment involving
$\fastrewritingconstraints$ is equivalent to entailment involving $\simplerewritingconstraints$.

In one direction, suppose that $I$ is a counterexample to entailment involving
$\simplerewritingconstraints$.
We fire the rules generating atoms $B_{e,f}$ to get instance $I'$.
We claim that the constraints of $\fastrewritingconstraints$ hold.
Clearly, the rules generating atoms $B_{e,f}$ hold.
Further, by construction, for any $e,f$ if $B_{e,f}$ holds exactly when
there is an annotation
We now consider the rule $B_{e_h, f_h}(\vec x) \rightarrow \exists \vec z ~ H(\vec z)$.
Considering a $\vec c$  such that $B_{e_h, f_h}(\vec c)$ holds, we want to claim
that there is an annotation $\annot$ such that $B_\annot(\vec c)$ holds. 

Recall that each $e_i, f_i$ is associated with some
variable $v$ that occurs as both $x_{e_i}$ and  $x_{f_i}$ in
$B(\vec x)$. If $B_{e_i, f_i}(\vec c)$ holds, we know that
either $c_{e_i}=c_{f_i}$ or $\iscrit(c_{e_i}) \wedge \iscrit(c_{f_i})$ holds.
If the latter happens, then we add the variable $v$ to our annotation.
We can then verify that $B_\annot(\vec c)$ holds.

Since we are assuming that the corresponding constraint of $\simplerewritingconstraints$ holds
in $I$, we can conclude that $I', \vec c \models \exists \vec z ~ H(\vec z)$.
From this we see that $I'$ is a counterexample to the entailment involving
$\fastrewritingconstraints$.

In the other direction, let $I'$ be a counterexample to the entailment for 
the constraints in $\fastrewritingconstraints$.
We claim that the constraints of $\simplerewritingconstraints$ hold of $I'$.
For constraints corresponding to source constraints
with  no repeated variables in the body, this is easy to verify,
so we concentrate on constraints deriving from source constraints
that do have repeated variables in the body.

Each of these constraints is of the form $B_\annot(\vec x) \rightarrow \exists \vec z ~ H(\vec z)$
for some annotation $\annot$. Fix a $\vec c$ such that $B_\annot(\vec c)$ holds.
We claim that $B_{e,f}(\vec c)$ holds for all $(e,f) \in P$. We prove this by induction on
the position of $(e,f)$ in the ordering of pairs in $P$.
Each  $(e,f)$ corresponds to some variable $v$ that is repeated.
If $v$ is in $\annot$, then $B_\annot(\vec c)$
implies  that $\iscrit(c_e) \wedge  \iscrit(c_f)$ hold. Using the corresponding
rule and the induction hypothesis we conclude that $B_{e,f}(\vec c)$ holds.
If $v$ is not in $\annot$ then $B_\annot(\vec c)$ implies that
$c_e=c_f$. Using the other rule generating $B_{e,f}$ in $\fastrewritingconstraints$,
as well as the induction hypothesis, we conclude that $B_{e,f}(\vec c)$ holds.
This completes the inductive proof that $B_{e,f}(\vec c)$ holds.
Now using the corresponding constraint of $\fastrewritingconstraints$ we
conclude that  $I', \vec c \models \exists \vec z ~ H(\vec z)$.
Since the constraints of $\simplerewritingconstraints$ hold, $I'$ is also
a counterexample to the entailment involving $\simplerewritingconstraints$.

\subsection{More details in proof of  Theorem \protect{\ref{thm:linlinupper}}: upper bounds
for $\ltgd$ source constraints and atomic maps} \label{app:linlinupper}
Recall the statement of Theorem \ref{thm:linlinupper}

\medskip

The problem $\discloseClass(\ltgd, \atomicmap)$
is in $\exptime$. If the arity of relations in the source schema is bounded, then the complexity drops to $\np$.
If further the query is atomic, the problem is in $\ptime$.

\medskip

We now give more details on the proof.
As mentioned in the body, is sufficient to get an $\exptime$ algorithm for the entailment problem produced by Theorem \ref{thm:fastcritinstrewriting}, since
then we can apply it to each $\secretquery_\annot$ in $\exptime$.
The constraints in $\critinstrewriteptime(\sourcecon) \cup \critinstrewriteptime(\M)$ are Guarded TGDs that are not necessarily
$\ltgd$s.
But the  bodies of these guarded
TGDs consist of a guard predicate and atoms over a  fixed ``side signature'',
namely the unary predicate $\iscrit$. 
We can apply now the \emph{linearization  technique}, originating
in  \cite{gmp} and refined in \cite{antoinemichael}. Given a
side signature $\sigside$ this is  an algorithm
that converts an entailment problem involving ta  set of  non-full $\incd$s
and Guarded TGDs using $\sigside$, producing an equivalent entailment problem involving the same query,
but only $\ltgd$s.
Further:
\begin{itemize}
\item The algorithm runs in $\exptime$ in general, and in $\ptime$ when the arity of the relations in the input is fixed
\item The algorithm does not increase the arity of the signature, and thus the size
of each output  $\ltgd$ is polynomially-bounded in the input.
\end{itemize}
See also Appendix G of \cite{antoinemichaelarxiv} for a longer exposition
of the linearization technique.
Thus for general arity, we can use this algorithm  to get an entailment problem with the same
query, a data set exponentially bounded in the input data $I'$
and a  set of  $\ltgd$s, each polynomially-sized in the inputs. By applying a standard
first-order query-rewriting algorithm to the query, we reduce this problem to evaluation
of a union of conjunctive queries get a UCQ  $Q'$ on $I'$. The size of
each conjunct in $Q'$ is  polynomially-bounded in the inputs, and so each conjunct $C$ can
be evaluated in time $|I'|^{|C'|}$, giving an $\exptime$ algorithm in total.

For fixed arity, we apply the same algorithm to get an entailment problem using $\incd$s of bounded
arity, which is known \cite{johnsonklug} to be solvable in $\np$. Further, when the query is atomic,
entailment with $\incd$s is in $\ptime$.

\subsection{Proof of Theorem \ref{thm:uidptime}: Disclosure for $\uid$ Source Constraints and $\projectionmap$  is {\PTime}}

We prove that when the source constraints are $\uid$s and the mappings
are projections, disclosure analysis is in $\ptime$. By Theorem
\ref{thm:critinst}, it suffices to show that the problem
$\hocwq(\critinst^\targetschema, \sourcecon \cup \viewcon(\M), \secretquery, \globalpreds(\M))$  is $\ptime$. We will thus first reduce this problem
a problem $\qentail(\instance,\Sigma,p)$ where $\Sigma$ is composed of
$\uid$ constraints and $\instance$ is composed of a single unary fact $\iscrit(\critelement)$.


\paragraph{Reachable predicates.}
We define the entailment graph over a set of $\incd$ constraints
$\Sigma$. In this graph, nodes correspond to predicates and there is
an edge $P\rightarrow R$ for each constraint $P(\vec{x})\rightarrow
R(\vec{y})$. Given an initial set of facts $\instance$, one can
compute the set $\reach(\Sigma,\instance)$ of entailed predicates. This
set is defined as the set of predicates reachable in the entailment
graph starting from the predicates appearing in $\instance$.

\paragraph{Visible position graph.}
In studying tuple-generating dependencies, one often associates a set
of dependencies with a graph whose edges represent the flow of data
from one relation to another via the dependencies.  See, for example
the position graph used in defining the class of weakly acyclic sets
of TGDs \cite{faginetaldataex}.

We develop another such graph, the \emph{visible position graph}
associated with a set of source constraints and mappings.  The nodes
are the pairs $(P,i)$ where $P$ is a predicate, $1 \leq i \leq ar(P)$
and there is an edge $(P,i)\rightarrow (R,j)$ when we have an $\incd$
(either a source constraint or a mapping rule) $P(\vec{x})\rightarrow
\exists \vec y~ R(\vec{t})$ with $x_i=t_j$.  We refer to a node in
this graph as a \emph{position}.  A position of a relation in the
source schema is said to be \emph{visible} if there is a path from
$(P,i)$ to a node $(R,j)$ such that $R$ belongs to the global
schema. Another other position is said to be \emph{invisible}. We see
that when a position $(P,i)$ is visible then for any fact $P(\vec{c})$
that holds in a possible world for $\hocwq(\critinst^\targetschema, \sourcecon \cup \viewcon(\M),
\secretquery, \globalpreds(\M))$ we must have $c_i=\critelement$.

Note that if we have $P(\vec{x})\rightarrow \exists \vec y~
R(\vec{t})$, $x_i$ is exported to $t_j$, and position $j$ of $R$ is
visible, then position $i$ of $P$ is visible as well.

\paragraph{Reduction to entailment.} Let $\Sigma=\sourcecon \cup \viewcon(\M)$ and 
$\instance_0=\critinst^\targetschema$.
We will reduce the problem $\hocwq(\instance_0, \Sigma, \secretquery, \globalpreds(\M))$ to the
problem $\qentail(\instance'_0,\tilde{\Sigma},\tilde{\secretquery})$,
where
$\tilde{\Sigma}=\Sigma_{reach}\cup\Sigma_1\cup\Sigma_{\critelement}$
is a set of $\uid$s, and $\tilde{\secretquery}$ is a CQ. Our reduction proceeds
as follows:
\begin{itemize}
\item We transform the  schema for sources creating a predicate $\tilde{P}$ for
  each source predicate $P$, where the arity of $\tilde{P}$ is the arity of $P$
  minus the number of positions $(P,i)$ that are visible.
\item $\instance'_0=\{\iscrit(\critelement)\}$.
\item $\Sigma_{reach}$ is built as the set of constraints
  $\iscrit(w)\rightarrow \exists \vec{x} ~ P(\vec{x})$ where $\vec{x}$ are fresh distinct
  variables and $P\in \reach(\hide_\M(\critinst^\targetschema),\Sigma)$.
\item $\Sigma_1$ is formed from the set of constraints $P(\vec{x})\rightarrow \exists \vec y ~ R(\vec{t})\in \Sigma$ 
such that there is an exported variable lying in an  invisible position of $P(\vec{x})$. For each such constraint,
$\Sigma_{1}$ contains the constraint
  $\tilde{P}(\vec{x}^*)\rightarrow \exists \vec{y^*} ~ \tilde{R}(\vec{t}^*)$ 
where $\vec{x}^*$
  denotes the projection of $\vec{x}$
  to the invisible positions of $P$, and similarly for $\vec{y^*}$ and $\vec{t}^*$.
\item $\Sigma_{\critelement}$ is formed from  constraints $P(\vec{x})\rightarrow \exists \vec y ~ R(\vec{t})\in \Sigma$
such that $P \in  \reach(\hide_\M(\critinst^\targetschema),\Sigma)$
and there is an exported variable $x$ lying in a visible position of $P(\vec{x})$, exported
to an invisible position of $R$. For each such constraint $\Sigma_{\critelement}$  includes
the constraint
  $\iscrit(x)\rightarrow \exists \vec {y}^* ~  \tilde{R}(\vec{t}^*)$ where
  $\vec{y}^*$ denotes the projection of $\vec{y}$
  to the invisible positions of $P$ and similarly for $\vec{t}^*$.
\item the query $\tilde{\secretquery}$ is built from $\secretquery$ by
  first replacing each conjunct $P(\vec{x})$ with its corresponding
  predicate $\tilde{P}(\tilde{\vec{x}})$, projecting out the visible
  positions. After this, for every variable $x$ that occurred in
  $\secretquery$ within both a visible and an invisible position, $x$ is replaced by $v$, while
  we add a conjunct $\iscrit(v)$.
\end{itemize}

\paragraph{Correctness of the reduction.}
The correctness of the reduction is captured in the following result:

\begin{proposition} \label{prop:reduceuid}
  For any source constraints $\sourcecon$ consisting of $\incd$s and
  $\M$ consisting of projection mappings, there is a disclosure over a
  schema $\sch$ with constraints $\sourcecon$ mappings $\M$ and secret
  query $\secretquery$ if and only if $\qentail(\instance'_0,
  \tilde{\Sigma},\tilde{\secretquery})$ holds.
\end{proposition}

\begin{proof}
We start with the argument for the left to right direction.  We let
$\instance'$ be a counterexample to the entailment
$\qentail(\instance'_0, \tilde{\Sigma},\tilde{\secretquery})$.  By
Theorem \ref{thm:chase}, we can assume that $\instance'$ is formed by
applying the chase procedure to $\instance'_0$.  In particular, each
fact in $\instance'$ can be assumed to use a predicate in
$\reach(\hide_\M(\critinst^\targetschema), \Sigma)$. 

We show that there is an instance $\instance$ that is a counterexample to
\[ \hocwq(\critinst^\targetschema, \sourcecon \cup \viewcon(\M), \secretquery, \globalpreds(\M))
\]
and thus (by Theorem \ref{thm:critinst}) we cannot have a disclosure.
We
form $\instance$ by filling out each visible position with $\critelement$.
We claim that $\instance$ satisfies each source constraint $\sigma= P(\vec{x})\rightarrow \exists \vec{y} ~ R(\vec{t})$.
Suppose that $P(\vec c)$ holds in $\instance$. Then $\tilde{P}(\vec c')$ holds
in $\instance'$, where $\vec c'$ projects $\vec c$
on to the invisible positions.
\begin{itemize}
\item 
First, suppose there is a variable $x$ in an invisible position of
$P(\vec x)$ exported to an invisible position in $R(\vec{t})$.  Then
since $\instance'$ satisfies $\Sigma_1$, we know that for some $\vec
d$, $\tilde{R}(\vec d)$ holds in $\instance'$, By the definition of
$\instance$, we have that $R(\vec d^*)$ holds, where $\vec d^*$ fills
out each visible position with $\critelement$.  We can see that
$R(\vec d^*)$ is the required witness for $P(\vec c)$.

\item
Next, suppose
there is a variable $x$ in a visible position $j$ of $P(\vec x)$
exported to an invisible position in $R(\vec{t})$.  Then we must have
$c_j=\critelement$. Since $P$ is in $\reach(\hide_\M(\critinst^\targetschema),
\Sigma)$ and $\instance'$ satisfies $\Sigma_{\critelement}$, we have
$\tilde{R}(\vec e)$ holding in $\instance'$ for some $\vec e$, and
hence $R(\vec f)$ holding in $\instance$ for some tuple where
$\critelement$ fills all the visible positions. Thus $\sigma$ holds in
this case as well.

\item Finally, note that a variable at an invisible position cannot be
  exported to a visible position. Therefore the only remaining case is
  the case where no variable has been exported. Since $P$ is
  reachable, then $R$ is also reachable therefore there is a
  constraints $\iscrit(x)\rightarrow \exists \vec{y}^* R(\vec{y}^*)\in
  \Sigma_{\reach}$ and thus $\tilde{R}(\vec d^*)$ holds in $\instance'$
\end{itemize}

We next claim that the image of $\instance$ under $\M$ agrees with $\critinst^\targetschema$.
\begin{itemize}
  \item
For every global schema predicate $G$, $G(\critelement \ldots
\critelement)$ occurs in the the image of $\instance$ under $\M$.
This follows easily from the fact that $\instance'$ contains
$\instance'_0$.

\item
If $G(\vec c)$ holds in the $\M$-image, then because each visible position
was filled out with $\critelement$, we must have each $c_i=\critelement$.
Thus the result follows.
\end{itemize}

Note that from the preceding claims, we know that $\instance$ is a
possible world for 
$\hocwq(\critinst^\targetschema, \sourcecon \cup \viewcon(\M), \secretquery,
\globalpreds(\M))$.  Finally, we claim that $\instance$ does not
satisfy $\secretquery$.

\begin{itemize}
\item
Suppose $\instance \models \secretquery$ with
homomorphism $h$ as a witness.  Since $\instance$ is a possible world
for $\hocwq(\instance_0, \Sigma, \secretquery, \globalpreds(\M))$, for
any variable $v$ occurring in a visible position,
$h(v)=\critelement$. Let $h'$ be formed from the restriction of $h$ to
variables that occur in $\tilde{\secretquery}$, by mapping the
additional variable $v$ to $\critelement$. Note that in $\instance'$,
$\iscrit(\critelement)$ holds.  For this, we see that $h'$ is a
homomorphism witnessing that $\instance' \models
\tilde{\secretquery}$.  This is a contradiction to the fact that
$\instance'$ is a counterexample to the entailment.
\end{itemize}

We now have argued that $\instance$ is a counterexample to $\hocwq(\instance_0, \Sigma, \secretquery, \globalpreds(\M))$, which completes the proof of the left to right direction.

For the other direction, suppose that $\instance$ is a counterexample
to $\hocwq(\instance_0, \Sigma, \secretquery, \globalpreds(\M))$. Note
that for any fact $R(\vec c)$ over the source relations in
$\instance$, for any visible position $i$ of $R$, we must have
$c_i=\critelement$.  Form $\instance'$ by projecting each fact in
$\instance$ to the invisible positions of the relation. We will argue
that $\instance'$ is a counterexample to the entailment produced by
the reduction.
\begin{itemize}
  \item
$\instance$ should contain $\iscrit(\critelement)$ therefore
    $\instance'$ extends $\instance'_0$.

  \item The fact that $\instance$ was a solution to
    $\hocwq(\instance_0, \Sigma, \secretquery, \globalpreds(\M))$ also
    guarantees that for all reachable predicates $P$ we have
    $\instance\models \exists \vec{x} ~ {P}(\vec{x})$ and thus
    $\instance'\models \exists \vec{x}^* ~ \tilde{P}(\vec{x}^*)$ and thus
    all constraints in $\Sigma_{\reach}$ are satisfied.

  \item Let us show that the constraints in $\Sigma_1$ are satisfied:
    fix a constraint $\sigma' \in \Sigma_1 = \tilde{P}(\vec x^*)
    \rightarrow \exists \vec y' ~ \tilde{R}(\tilde{\vec{t}})$, derived
    from source constraint $\sigma= P(\vec{x})\rightarrow \exists
    \vec{y} ~ R(\vec{t})$.  Fix a fact $F' =\tilde{P}(\vec c^*)$ in
    $\instance'$. By definition of $\instance'$, $\vec c^*$ extends to
    a $\vec c$ satisfying $P$ in $\instance$. Thus, since $\instance
    \models \Sigma$, there is a fact $G=R(\vec d)$ that holds in
    $\instance$ with $d_i=c_j$ whenever $t_i=x_j$.  We can project to
    the invisible positions to get a fact $G'=\tilde{R}(d_{j_1} \ldots
    d_{j_n})$ in $\instance'$.  We claim that $G'$ is a witness for
    the satisfaction of $\sigma'$ with respect to $F'$.  Consider any
    variable $x$ exported from $F'$ to position $j'$ of $G'$ where $x$
    is mapped to value $c$ in $\vec c^*$.  Then in $\sigma$, $x$ was
    exported to the corresponding invisible position $j$ in
    $R(\vec{y})$, and from this we see that $d_j=c$ as required.
  \item 
    Now consider a constraint $\sigma' \in \Sigma_{\critelement} =
    \iscrit(x)\rightarrow \exists \vec y^* ~
    \tilde{R}(\tilde{\vec{t}^*})$. Since $\iscrit(x)$ holds only for
    $x=\critelement$ in $\instance$, we only have to verify that
    $\tilde{R}(\vec{e}^*)$ holds for some $\vec e$ such that
    $e_\ell=\critelement$ (where $\ell$ is the position of $x$ in
    $\vec y^*$).  Let us suppose that $\sigma'$ was derived from
    source constraint $\sigma= P(\vec{x})\rightarrow \exists \vec{y} ~
    R(\vec{t})$ where $j$ is the position of the exported in $\vec{y}$
    and $i$ is the position of the exported variable in $\vec{t}$. By
    the definition of $\Sigma_{\critelement}$, we know that $P$ is a
    reachable predicate, and hence $P(\vec d)$ must hold for some
    $\vec d$ in $\instance$ and since $d_j$ is visible we have
    $d_j=\critelement$. Because $\instance\models \sigma$ we have
    $\tilde{R}(\vec e)$ holds in $\instance$ for some $\vec e$ such
    that $e_i=\critelement$ and thus $R(\vec e^*)$ is the required
    witness for $\sigma'$.

\item 
Finally, we argue that $\instance'$ does not satisfy  $\tilde{\secretquery}$.
Suppose by way of contradiction that $\instance'$ satisfies $\tilde{\secretquery}$ via
homomorphism $h'$. Note that the  variables of $\secretquery$ that
do not occur in $\tilde{\secretquery}$ are those that occur
only in visible positions within an atom of $\secretquery$.
We extend $h'$ to a mapping $h$ from the  variables of $\secretquery$  to $\instance$ by mapping each
such  variable $x$ to $\critelement$. We argue that
$h$ is a homomorphism of $\secretquery$ to $\instance$.
Consider  an atom $R(\vec t, \vec t')$ of $\secretquery$, where $\vec t$ correspond
to the invisible positions.  Suppose first
that the corresponding atom of $\tilde{\secretquery}$ is of the form
  $\tilde{R}(\vec t^*)$ where $\vec t^*$ is obtained from $\vec t$ by
replacing  any variable shared with a visible position by $v$.
We know that $\tilde{R}(h(t^*_1) \ldots h(t^*_j))$ holds in $\instance'$
because $h$ is a homomorphism. Thus $R(h(t^*_1), \ldots h(t^*_j), \vec e)$ holds
in $\instance$ for some $\vec e$. By the properties of visible positions and the fact
that $\instance$ is a possible world for 
$\hocwq(\instance_0, \Sigma, \secretquery, \globalpreds(\M))$, we see
that each $e_i=\critelement$. Thus $h$ not only preserves the  atom $\tilde{R}(\vec t^*)$, but
it also preserves the additional atom $\iscrit(v)$, since $\iscrit(\critelement)$ holds
in $\instance'$.  thus $h$ is a homomorphism, contradicting the fact that $\instance$ is a counterexample to 
$\hocwq(\instance_0, \Sigma, \secretquery, \globalpreds(\M))$. 
\end{itemize}

Since $\instance'$ extends $\instance_0$, satisfies the constraints $\tilde{\Sigma}$, and
does not satisfy the query $\tilde{\secretquery}$, it is a counterexample to the entailment, completing
this direction of the argument.
\end{proof}

\paragraph{Overview of {\PTime} algorithm for entailment with $\uid$s over a single fact.}
At this point we have restricted to a CQ entailment problem for a set
of $\uid$s and a single fact. It was claimed in \cite{uidrewrite} that
there is a polynomial time query rewriting for $\uid$s, and from this
it would easily follow that our entailment problem is in $\ptime$
(query evaluation is $\ptime$ when these is a single fact). However
later work (footnote on page 38 of \cite{kikotjacm}) refers to flaws
in this argument, and says that polynomial rewritability is open.  We
therefore give a direct proof that such an entailment problems are in
$\ptime$. This will proceed via several steps:
\begin{itemize}
\item A reduction to the case of  ``binary schemas'': those where
the arity of each predicate is at most $2$.
\item Query simplification, which will reduce the query to a connected
acyclic query.
\item Reduction to atomic entailment.
\end{itemize}

\paragraph{Reduction to binary schemas.} We begin by using verbatim
an idea of \cite{uidrewrite}, reducing to the same problem but when
the input schema is binary. We do this via a standard reduction of
general arity reasoning to binary reasoning, introducing predicates
$R_i(t, v)$ for every relation $R$ of arity $n\geq 1$ and each $1 \leq
i \leq n$; informally these state that $v$ is the value in position
$i$ of $n$-tuple $t$. We also introduce a predicate $R_{\exists}(t)$
for each predicate $R$; informally this states that there is some
tuple $t$ in the predicate $R$. We translate each $\uid$ $B(\vec x)
\rightarrow \exists \vec y ~ H(\vec t)$ exporting a variable $x_i$
from position $i$ to position $j$ to a $\uid$ $B_i(t, x_i) \rightarrow
\exists t' ~ H_j(t',x_i)$. For each $\uid$, $H(\vec x) \rightarrow
B(\vec y)$ that is not exporting a variable, we create a rule
$H_{\exists}(t)\rightarrow \exists t' ~ B_{\exists}(t')$. We also
create rules $R_i(t,x)\rightarrow R_{\exists}(t)$ and
$R_{\exists}(t)\rightarrow \exists x~ R_i(t,x)$ for each predicate $R$
and $1\leq i\leq n$ where $n$ is the arity of $R$.  Finally the query
$\secretquery$ is transformed into $\secretquery'$ where each conjunct
$R(x_1,\dots,x_n)$ is transformed into the conjunction
$R_1(t,x_1)\land \dots \land R_n(t,x_n)\land R_{\exists}(t)$, for a
fresh variable $t$. Finally the database over the binary schema is built in the
following way: for each fact $R(\vec v)$ of the initial database, we
create a fresh value $t$ and we add the conjunct $R_i(t,x)$ for $1\leq
i\leq n$ where $n$ is the arity of $R$ and we also add
$R_\exists(t)$. Further details can be found in
\cite{uidrewrite}. Note that, in the resulting problem, each
frontier-0 rule produced has a body with an atom over a unary
predicate.
\begin{proposition} \label{prop:reducebinary} The transformation above
preserves query entailment.
\end{proposition}

\paragraph{Special form of the chase: annotated chase forest.} In the case of $\uid$s the chase
process applied to our single-fact instance $\instance_0$ produces an
in instance $\chase_\Sigma(\instance_0)$ that will be
infinite. However, it has a special shape that we can exploit. For
the remainder of this section, by $\chase_\Sigma(\instance_0)$ we consider an instance
formed from a \emph{restricted chase sequence}, in which a witness to
a TGD $\phi(\vec x) \rightarrow \exists \vec y ~ H(\vec t)$ is added
to instance $\instance_i$ for binding $\vec c$ to $\vec x$ only if
$\instance, \vec c \models \phi(\vec x) \wedge \neg \exists \vec y ~
H(\vec t)$.  It is known \cite{faginetaldataex} that in Theorem
\ref{thm:chase} it suffices to consider such instances.  The
\emph{annotated chase} is a node- and edge-labelled forest formed from
$\chase_\Sigma(\instance_0)$ as follows:
\begin{compactitem}
\item  the nodes are the values  of $\chase(\instance)$
\item the node label of a value $v$ is  the collection of unary predicates holding at $v$
\item an edge labeled by fact $F$ mentioning $v_1$ and  $v_2$  connects a value $v_1$ to a value $v_2$ if
$F$ holds in $\chase(\instance)$ and $v_2$ is generated in the chase step 
that produces $F$.
\end{compactitem}
We can see that this graph is a forest where 
he roots are $\critelement$ (the value where
$\iscrit(\critelement)$ holds) as well as some other trees rooted to
reachable facts generated from frontier-$0$ dependencies and thus
rooted at elements $t$ where $R_{\exists}(t)$ holds for some $R$.
Further, since the chase is restricted, we can see that this graph has
the \emph{unique adjoining label property}: for each $v_1$, for each
predicate $P$, there cannot be two nodes $v_2, v'_2$ adjacent to $v_1$
such that the edge $e$ from $v_1$ to $v_2$ and $e'$ from $v_1$ to
$v_2$ both are labelled with the same predicate and have $v_1$ in the
same position. Furthermore, the restricted chase also ensures that the
forest is composed of at most one tree per predicate since all the
roots that are produced needs to be different.

\paragraph{First query simplification: eliminating forking pairs.}
Given a CQ $Q$,
 a pair of distinct atoms $A_1$ and $A_2$ sharing the same predicate and a
variable at the same position (i.e. $q_1=R(x,z)$ and $q_2=R(x,y)$ or
$q_1=R(z,x)$ and $q_2=R(y,x)$) is a \emph{forking pair} of $Q$. 
We say that a query $Q$ is \emph{non-forking} when there are no forking pairs.

\begin{proposition} \label{prop:nonforking}
If a CQ $Q$ has a forking pair
$A_1=R(x,z)$ and $A_2=R(x,y)$ and $Q'$ is the query $Q$ where the variable
$z$ is  replaced with $y$, then
$\qentail(\instance_0,\Sigma,Q)=\qentail(\instance_0,\Sigma,Q')$
\end{proposition}
\begin{proof}
Let $\instance=\chase_\Sigma(\instance_0, \Sigma)$. If $\secretquery'$ holds
in $\instance$, then clearly the same holds of $\secretquery$. Conversely
suppose $\secretquery$ holds in $\instance$ via homomorphism $h$, and suppose $h(y) \neq h(z)$.
 This gives us a violation of the unique adjoining label property.
\end{proof}

Applying the proposition  above, we can assume that $Q$ is non-forking.
Without loss of generality, we can also assume that $Q$ is connected
(otherwise we can test the entailment of each connected part).



\paragraph{Second simplification: reducing to acyclic queries} The \emph{CQ-graph.} of a CQ $Q$ is the node-
and edge-labelled graph whose nodes are the variables of $Q$ and whose
edges are labelled with atoms of $Q$ such that:
\begin{itemize}
\item an edge between variables labelled with $x$ and $y$ is labelled
with the binary atoms containing both $x$ and $y$;
\item a node $x$ is labelled with the set of unary predicates in $Q$
  containing $x$.
\end{itemize}

The CQ-graph said embedded in some annotated chase forest $T$ if there
is a homomorphism $h:\A\rightarrow T$ preserving edges, i.e. if there
is an edge $x$ to $y$ labeled with $R(a,b)$ then $T$ should contain an
$\A(x)$ to $\A(y)$ labeled with $R(\A(a),\A(b))$ and nodes, i.e. if
there is a predicate $P(x)$ on the node $x$ then there should be
$P(\A(x))$ in $T$. The homomorphism $h$ is called an embedding of $Q$
in $T$.

It is immediate from the completeness of the chase procedure
that for any annotated chase forest $T$ for $\chase_\Sigma(\instance)$, a query is entailed if and only
if its CQ-graph is embedded in $T$.
Our reduction to the case of a CQ with acyclic CQ-graph will depend heavily on the following
observation:

\begin{proposition} \label{prop:injective}
Any embedding of a connected and non-forking CQ $Q$ into an annotated chase
forest for $\chase_\Sigma(\instance_0)$ must be injective.
\end{proposition}

\begin{proof}
Let $Q$ be a connected and non-forking and let $h$ be an embedding. Let
us prove by induction on the size of the path between $x$ and $y$ that
$h(x)\neq h(y)$ when $x\neq y$.

Two neighboring nodes cannot be sent to the same value. For a path of
size $2$, if we have $z$ such that $x,z,y$ forms a path in the
CQ-graph of $Q$ then $h(x)$ has to be different than $h(y)$ otherwise
the label from $x$ to $z$ and from $z$ to $y$ would be the same and
there would be a forking pair in $Q$.

Let $x=p_1, p_2, \dots, p_k=y$ with $k \geq 4$ be a path in the CQ-graph
between $x$ and $y$. By induction the $h(p_i)$ for $i<k$ are all
distinct and thus the distance between $h(x)$ and $h(p_{k-1})$ is at
least $k-2$ hence $h(y)$ is at least at distance $k-3>0$ of $h(x)$.
\end{proof}

Our reduction to the acyclic case  follows immediately:

\begin{corollary} \label{cor:acyclic}
If a connected non-forking CQ $Q$ is entailed by $\Sigma$ over then  the CQ-graph of 
$Q$ is acyclic.
\end{corollary}

\begin{proof}
$Q$ is entailed 
The image of the CQ-graph through the injective homomorphism
is a forest.
\end{proof}

\paragraph{Determining entailment for acyclic connected graphs.} We now give the final setp
in our algorithm, which deals with deciding entailment of a connected, non-forking query
$Q$, which by Corollary \ref{cor:acyclic} must have an acyclic CQ-graph. 
Given an acyclic connected undirected graph and any vertex $v$ of the graph,
we can direct it be a tree with $v$ as the root.
Thus for such a $Q$  having $n$ variables,
the \emph{tree
  arrangements} are the $n$ possible ways to root the CQ-graph
of the query $Q$.  We are particularly interested in arrangements of $Q$
where the directionality from parent to child reflects the entailment structure
relative to $\Sigma$ between atoms in the query.
A tree arrangement $\A$ of $Q$ is \emph{faithfully entailed} if
for every variable $y$ in $Q$ with parent $x$ in the tree, 
there is an atom $A$ containing $x$ and not containing $y$ such that 
$A \wedge \Sigma$ entails $\exists y ~ B_{x,y}$, where $B_{x,y}$ is the conjunction
of all atoms whose variables are contained in $\{x, y\}$; in the case that $y$ is the
root, we require $\Sigma$ alone to entail $\exists y ~ B_{x,y}$.

In a faithfully-entailed tree arrangement, the conjunction of atoms holding
at the root of the tree entails the existence of the whole tree. We can further
find a single atom that entails the whole tree.  A \emph{root-generating atom}
of a tree arrangement is  an atom $A$ (not necessarily
in $Q$) containing the root variable $r$, such
that $A \wedge \Sigma$  generates all atoms mentioning $r$.

\begin{proposition}
A faithfully entailed tree arrangement for $Q$ must have a root-generating atom.
\end{proposition}
\begin{proof}
We know that $Q$ must hold in the chase of the initial fact under $\Sigma$, and by Proposition
\ref{prop:injective} we know that there is an injective homomorphism $h$ from
$Q$ to the chase. Consider the point in the chase process
where  value $h(r)$ is first generated. This occurs by firing some rule with an atom, where
the head has either  a binary atom $A(x,y)$ or a unary atom $B(x)$. We consider
the case where the atom is binary,  and where the generated atom is $A(h(r),s)$. In this
case the fact $A(h(r),s)$ must generate every fact containing  $r$. Thus we can take the atom $A(r, w)$, where
$w$ is a fresh variable, as a root-generating atom.
The case of unary atoms and the case where $r$ is in the second position of the fact is similar.
\end{proof}

Given a tree arrangement $T$ of $Q$ and variable $x$ of $Q$, $T_x$ denotes
the the restriction of $T$ to the variables that are descendants of $x$ in $T$.




The main idea of our $\ptime$ algorithm is that it suffices to descend through the tree arrangement,
checking some  entailments for each parent-child pair in isolation.

\begin{proposition} \label{prop:checkrealizeinductive}
  There is a $\ptime$ algorithm taking as input a variable $x$ in a CQ  $Q$, a tree arrangement of $Q$,
and  an atom $A$ containing $x$ such that the existential quantification of $A$ is entailed by
$\Sigma$,
and  determining whether  $T_x$ is faithfully entailed
and $A$ is a root-generating atom.
\end{proposition}
\begin{proof}
We first check whether  $A$ is a root-generating atom,
  using a $\ptime$ inference algorithm  for $\uid$s \cite{uidptimeimp}.
We then consider each child $y$ of $x$ in the tree arrangement.  We know that there
is exactly one conjunct $B$ containing $x$ and $y$.  We check
whether $A$ entails $\exists y ~ B$, and then call the algorithm
recursively for $y$ and $B$. If each recursive call succeeds, the
algorithm succeeds.
\end{proof}

From the prior proposition we get a $\ptime$ algorithm for the arrangement as a whole:
\begin{proposition} \label{prop:checkrealizetop}
  There is a $\ptime$ algorithm taking a tree arrangement of CQ  $Q$, 
and  an atom $A$ containing the root of the arrangement,
and  determines whether  the whole tree arrangement  can be faithfully entailed and
$A$ is a root-generating atom.
\end{proposition}
\begin{proof}
We first need to
check that $A$ is entailed, which amounts to 
checking that $\Sigma\models \iscrit(c)\rightarrow \exists y ~ A$. As before
this
can be done  using \cite{uidptimeimp}.
We then utilize the algorithm of Proposition \ref{prop:checkrealizeinductive}.
\end{proof}
Note that Proposition \ref{prop:checkrealizetop} gives a polynomial time algorithm
for checking whether a tree arrangement can be faithfully entailed. We can apply
the algorithm of the proposition with
every possible unary and binary atom $A$ containing the root variable. In the binary case, we consider all atoms containing
the root variable  and an additional
fresh variable.

\paragraph{Putting it all together.}
Putting together our reduction to $\uid$-entailment (Proposition \ref{prop:reduceuid}), our
 schema simplification (Proposition \ref{prop:reducebinary})
the query simplifications (the reduction to connected CQs,
Proposition \ref{prop:nonforking}, and Corollary \ref{cor:acyclic}), and our $\ptime$ algorithm for simplified queries 
(Proposition \ref{prop:checkrealizetop})
we obtain the proof of Theorem \ref{thm:uidptime}.

\newpage
\section{Detailed Proofs from Section \ref{sec:lower}: Lower Bounds for Disclosure}
\subsection{Proof of the first part: Theorem \ref{thm:incdguardedlower}:
 $\twoexp$-hardness for $\incd$ and $\guardedmap$ without arity bound}

Recall the first part of Theorem \ref{thm:incdguardedlower}:
\medskip

 $\discloseClass(\incd,\guardedmap)$ is $\twoexp$-hard.

\medskip

Recall that Theorem \ref{thm:critinst} relates disclosure to a $\hocwq$ problem on a very simple instance.
Also recall from Section \ref{sec:reduction} the intuition that such a problem amounts to a classical entailment problem 
for a CQ  over a very simple
instance,  using the source dependencies and $\singleconstantrule$s: of the
form  $\phi(\vec{x}) \rightarrow x=\critelement$, where $\phi$  will be the body of a mapping. 
We show here how to simulate the run of an alternating $\expspace$ Turing machine $\TuringMachine$ 
without explicitly
using $\singleconstantrule$s, instead using  inclusion dependencies as source constraints
coupled with guarded mappings.
An alternating Turing machine $\TuringMachine$ is a $6$-tuple 
$(\states,\letters,\transitionFunction_\alpha,\transitionFunction_\beta,\initialState,\stateType)$ where:
\begin{itemize}
 \item $\states$ is the finite set of states
 \item $\letters$ is the finite tape alphabet
 \item $\transitionFunction_\alpha$ and $\transitionFunction_\beta$ are functions from $\states\times\letters$ to $\states\times\letters\times\{L,R\}$
 \item $\initialState \in \states$ is the initial state
 \item $\stateType$ is a function from $\states$ to $\{accept,reject,\forall,\exists\}$ that specifies the type of each state.
\end{itemize}
We assume that $\TuringMachine$ always alternates between existential and 
universal states, and that there is a unique final state, 
that can be reached only if the head is in the first cell and contains a 
specific symbol. All of these assumptions can be made without loss of 
generality. If $\TuringMachine$ is in a configuration where whose state $q$ 
is such that $\stateType(q) = accept$, the configuration is said 
to be accepting. 
If $\TuringMachine$ is in a configuration where whose state $q$ is such 
that $\stateType(q) = \forall$, the configuration is said to be
accepting if its $\alpha$ and $\beta$ successors (obtained after applying $\transitionFunction_\alpha$ or $\transitionFunction_\beta$) are accepting. 
If $\TuringMachine$ is in a configuration whose state $q$ is such 
that $\stateType(q) = \exists$, the configuration is said to
be accepting if its $\alpha$-successor or its $\beta$-successor is accepting. 
A more thorough introduction to Turing machines can be found
 in \cite{papadimitriou}.  

We first present the reduction, and show its correctness in the next subsection.

\subsection{The Reduction}
We will create constraints and mappings that will serve to perform the following tasks:
\begin{itemize}
 \item generate addresses for cells of $\TuringMachine$ in such a way that one can check whether two addresses are consecutive in a 
guarded way. The same addresses will be used for all the configurations. This will be done by a 
mapping creating $k$ copies of two individuals that represent $0$ and $1$, along with inclusion dependencies that perform permutations and generate $2^k$ addresses;
 \item encode the content of a cell, the position of the head, and the state of the head: for each cell, 
we store a vector whose length is the size of $(\letters \cup \{\flat\}) \times (\states\cup\bot)$. 
Each position corresponds to an element $(l,s)$ of that set; we will arrange that the position
 contains $\critelement$ if and only if the cell contains $l$, and either the head is over that cell and is in state $s$, or the head is 
not over that cell and $s = \flat$. All values are first freshly instantiated by inclusion dependencies, and mappings
 are then responsible for unifying the correct positions with $\critelement$;
 \item ensure that the tape that is associated with a successor of a configuration can be obtained by a transition of the Turing machine: 
this is also performed by using a mapping to enforce the correct positions of the cell to be unified with $\critelement$;
 \item check that configurations are accepting: this is the case either when the corresponding tape is in a final accepting state, or when it is in an existential state and one of the two successor configurations is accepting, or it is in a universal state, and both successor configurations are accepting.
\end{itemize}
\begin{figure*}[h]
\begin{center}
 
\begin{tikzpicture}
\draw[rounded corners=3mm] (4,5)
  -- (4.55,5)
  -- (4.55,-0.15)
  -- (0,-0.15)  
  -- cycle;

\draw (3.6,3.85) node (c) {$c$};
\draw (4.3,3.85) node (ac) {$ac$};
\draw (3.75,0.25) node (beta) {$\beta~ac_\beta$};
\draw (1,0.25) node (alpha) {$\alpha~ac_\alpha$};
\draw (3,2.8) node (y00) {$y_0$};
\draw (3,2.1) node (y01) {$y_0$};
\draw (3,1.4) node (y10) {$y_1$};
\draw (3,0.7) node (y11) {$y_1$};
\draw (2.4,0.25) node (children1) {$\children_\forall$};

\draw[rounded corners=6mm] (1.5,1.5) 
  -- (-3.6,-3.6)
  -- (1.5,-3.6) 
  -- cycle;

 \draw (0.6,-3.2) node (alphaAlpha) {$\alpha_\beta~ac_{\alpha_\beta}$};
\draw (-2.3,-3.2) node (alphaAlpha) {$\alpha_\alpha~ac_{\alpha_\alpha}$};
\draw (0,-0.5) node  {$y_0$};
\draw (0,-1.2) node  {$y_0$};
\draw (0,-1.9) node {$y_1$};
\draw (0,-2.6) node {$y_1$};
\draw (-1,-3.25) node (children2) {$\children_\exists$};

\draw[dashed, rounded corners=2.5mm] (3.8,4.1)
  -- (4.2,3.5)
  -- (8.5,3.5)
  -- (8.5,3)
  -- (4.2,3)
  -- (4.2,0)
  -- (3.3,0)
  -- (3.3,4.1)
  -- cycle;
\draw (9,3.25) node {$\data$};
\draw (6.5,3.25) node {$y_0~y_0\quad\mathbf{v^{00}}\mathbf{v_{prev}^{00}}~\mathbf{v_{next}^{00}}$};

\draw[dotted, rounded corners=2.5mm] (3.8,4.1)
  -- (4.2,3.5)
  -- (4.2,2.75)
  -- (8.5,2.75)
  -- (8.5,2.25)
  -- (4.2,2.25)
  -- (4.2,0)
  -- (3.3,0)
  -- (3.3,4.1)
  -- cycle;
\draw (6.5,2.5) node {$y_0~y_1\quad\mathbf{v^{01}}~\mathbf{v_{prev}^{01}}~\mathbf{v_{next}^{01}}$};
\draw (9,2.5) node {$\data$};  

  \draw[densely dotted, rounded corners=2.5mm] (3.8,4.1)
  -- (4.2,1.95)
  -- (8.5,1.95)
  -- (8.5,1.45)
  -- (4.2,1.45)
  -- (4.2,0)
  -- (3.3,0)
  -- (3.3,4.1)
  -- cycle;

\draw (6.5,1.7) node {$y_1~y_0\quad\mathbf{v^{10}}~\mathbf{v_{prev}^{10}}~\mathbf{v_{next}^{10}}$};
\draw (9,1.7) node {$\data$};

  \draw[loosely dotted, rounded corners=2.5mm] (3.8,4.1)
  -- (4.2,1)
  -- (8.5,1)
  -- (8.5,0.5)
  -- (4.2,0.5)
  -- (4.2,0)
  -- (3.3,0)
  -- (3.3,4.1)
  -- cycle;

\draw (6.5,0.75) node {$y_1~y_1\quad\mathbf{v^{11}}~\mathbf{v_{prev}^{11}}~\mathbf{v_{next}^{11}}$};
\draw (9,0.75) node {$\data$};

\draw (5.1,0.25) node {$\underbrace{\qquad}$};
\draw (5.1,0) node {$\mathrm{address}$};
\end{tikzpicture}
\end{center}

  \caption{The generated structure} \label{fig:coding}
\end{figure*}
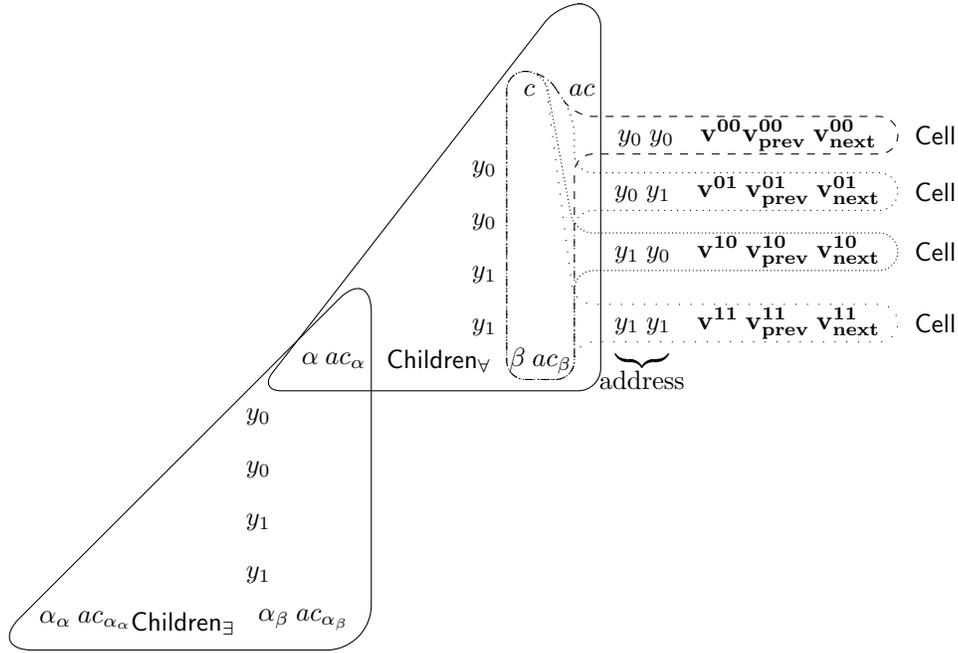

 Let us describe the source signature. For each predicate, we will explain  what feature of the ATM $\TuringMachine$ 
it should represent in the appropriate instance generated by the constraints.
By ``the appropriate instance'', we mean the visible chase of the initial instance
over the source constraints and mappings: this was introduced after Theorem \ref{thm:chase}, and it was noted
that it is the canonical instance for the source and targets  to consider for  disclosure.

 We use $\mathbf{y^{1,k}}$ to represent a tuple $(y^1,\ldots,y^k)$, and $\mathbf{y^{k}}$ to represent the tuple $(y,\ldots,y)$ of size $k$. 
 \begin{itemize}
  \item $\children_\forall(c,c_\alpha,c_\beta,ac,ac_\alpha,ac_\beta,\mathbf{y_0^{1,k}},\mathbf{y_1^{1,k}},r,z,y_0,y_1)$.
The intended meaning is that a configuration $c$ is universal and has as children $c_\alpha$ and $c_\beta$, and that the acceptance bit of $c$ is $ac$, of $c_\alpha$ is $ac_\alpha$ and of $c_\beta$ is $ac_\beta$. The last four positions are placeholders: $r$ for the root of the tree of configurations, $z$ for $\critelement$, $y_0$ for a value representing $0$ and $y_1$ for a value representing $1$. 
  \item $\children_\exists(c,c_\alpha,c_\beta,ac,ac_\alpha,ac_\beta,\mathbf{y_0^{1,k}},\mathbf{y_1^{1,k}},r,z,y_0,y_1)$: same intended meaning, except that $c$ is existential. 
  \item $\data(c_p,c_n,\mathbf{y^{1,k}},\mathbf{v},\mathbf{v_{prev}},\mathbf{v_{next}},r,z,y_0,y_1)$ 
with intended meaning that the cell of address $\mathbf{y^{1,k}}$ of the tape represented by $c_n$ has a content represented 
by $\mathbf{v}$, while the previous cell has a content represented by $\mathbf{v_{prev}}$ and the next cell has  content represented 
by $\mathbf{v_{next}}$. The last four positions are placeholders for the root of the tree of configurations, $\critelement$, 
a value representing $0$ and a value representing $1$.
  \item $\data_i^{c}(c,\mathbf{y^{1,k}},x,z,y_0,y_1)$ with intended meaning that the cell of address $\mathbf{y^{1,k}}$ in configuration $c$ contains $x$ at the $i^\mathrm{th}$ position of the representation of its content. $\data_i^{p}$ and $\data_i^{n}$ play similar roles for the cell before and after the cell of address $\mathbf{y^{1,k}}$.
  \item $\genAddr$ is an auxiliary predicate used to generated an exponential number of addresses.
  \item $\kw{succ}_\alpha(c_p,c_n)$ states that $c_n$ is the $\alpha$-successor of $c_p$ (and similarly for $\beta$)
 \end{itemize}

Below we will always use the symbol $\quantif$ to range over $\{\forall,\exists\}$.
 
The structure generated by the inclusion dependencies is represented Figure \ref{fig:coding}. Atoms are represented by geometric shapes in the inside of which are arguments (some are omitted to ease the reading). The $\children_\quantif$ atoms form a tree shaped structure, and induce a tree structure on the configuration identifiers: for instance, $c$ is the parent of $\alpha$ and $\beta$. $\data$ atoms are associated with a configuration identifier (for instance, those represented are associated with $\beta$), and has the parent configuration identifier to ensure guardedness of the mappings used in the following reduction. Note that the elements used to describe the cell's addresses ($y_0$ and $y_1$) also appear in the $\children_\quantif$ atoms, to ensure guardedness.

\paragraph{Initialization.}
We first define a mapping $\T_{\init}(x)$, introducing some elements in the visible chase. The   definition of this mapping is:
\begin{align*}
 \children_\exists(c_{root},&c_\alpha,c_\beta,ac_{root},ac_\alpha,ac_\beta,\\
 &\mathbf{y_0^{1,k}},\mathbf{y_1^{1,k}},c_{root},x,y_0^0,y_1^0)
\end{align*}
\paragraph{Generation of the tree of configuration.}
$\alpha$-successors have themselves $\alpha$- and $\beta$-successors, and are existential if their parent is universal:
\begin{align*}
 &\children_\forall(c,\alpha,\beta,ac,ac_\alpha,ac_\beta,\mathbf{y_0^{1,k}},\mathbf{y_1^{1,k}},r,z,y_0,y_1) \\
 &\rightarrow \exists \alpha_\alpha,\alpha_\beta,ac_{\alpha_\alpha},ac_{\alpha_\beta}\\
 &\children_\exists(\alpha,\alpha_\alpha,\alpha_\beta,ac_\alpha,ac_{\alpha_\alpha},ac_{\alpha_\beta},\mathbf{y_0^{1,k}},\mathbf{y_1^{1,k}},r,z,y_0,y_1)
\end{align*}
And similarly for $\children_\exists$ and for the $\beta$-successor.

\paragraph{Universal and existential acceptance condition.}
If both successors of a universal configuration $n$ are accepting, so is $n$. We create a mapping $\T_\forall(x)$ with definition:
\begin{align*}
 \children_\forall(c,\alpha,\beta,x,z,z,\mathbf{y_0^{k}},\mathbf{y_1^{k}},r,z,y_0,y_1) \\
\end{align*}

If the $\alpha$-successor of an existential configuration $n$ is accepting, so is $n$. We create a mapping $\T_{\exists,\alpha}(x)$ with definition:
\begin{align*}
 \children_\exists(c,\alpha,\beta,x,z,ac_\beta,\mathbf{y_0^{k}},\mathbf{y_1^{k}},r,z,y_0,y_1) \\
\end{align*}

We create a similar mapping $\T_{\exists,\beta}$ for the $\beta$-successor.

\paragraph{Tape representation and consistency of tapes.}
We now focus on the representation of the tape and its consistency. We generate $2^k$ addresses and associated values:
\begin{align*}
 &\children_{\quantif}(c,\alpha,\beta,ac,ac_\alpha,ac_\beta,\mathbf{y_0^{1,k}},\mathbf{y_1^{1,k}},r,z,y_0,y_1)\\
 &\rightarrow \genAddr(c,\alpha,\mathbf{y_0^{1,k}},\mathbf{y_1^{1,k}},r,z,y_0,y_1) 
 \end{align*}
 
 \begin{align*}
 &\children_{\quantif}(c,\alpha,\beta,\mathbf{y_0^{1,k}},\mathbf{y_1^{1,k}},r,z,y_0,y_1)\\
 &\rightarrow \genAddr(c,\beta,\mathbf{y_0^{1,k}},\mathbf{y_1^{1,k}},r,z,y_0,y_1) 
 \end{align*}

$\genAddr$ will generate addresses to represent the tape associated with its fifth argument. To emphasize
this, we use the letter $n$ (as node) at this position, while the fourth argument contains its parent configuration, denoted by $p$.
\begin{align*}
 &\genAddr(c_p,c_n,a_1,\ldots,a_i,\ldots,a_{k+i},\ldots,a_{2k},r,z,y_0,y_1) \\
 &\rightarrow \genAddr(c_p,c_n,a_1,\ldots,a_{k+i},\ldots,a_{i},\ldots,a_{2k},r,z,y_0,y_1) 
 \end{align*}
 
 For each address, we initialize its content (as well as the content of the previous and next cells) by fresh values $\mathbf{v},\mathbf{v_{prev}},\mathbf{v_{next}}$.
\begin{align*}
 &\genAddr(c_p,c_n,a_1,\ldots,a_{2k},z,y_0,y_1) \rightarrow \exists \mathbf{v},\mathbf{v_{prev}},\mathbf{v_{next}}~ \\
 & \data(c_p,c_n,a_1,\ldots,a_{k},\mathbf{v},\mathbf{v_{prev}},\mathbf{v_{next}},r,z,y_0,y_1) 
 \end{align*}
Note that the values $\mathbf{v}$, $\mathbf{v_{prev}}$ and $\mathbf{v_{next}}$ are vectors of length the size of $(\letters \cup \{\flat\}) \times (\states \cup \bot)$. In particular, we use the notation $\mathbf{l_i}(x)$ to represent a vector of same length, composed of fresh variables, except for the position $i$, that contains $x$.

We now use mappings to force some of these values to be equal to $\critelement$. Each position of $\mathbf{v}$ represents 
an element of $(\letters \cup \{\flat\}) \times (\states\cup\bot)$, and we will enforce exactly one of these positions to contain $\critelement$. If the head of the Turing machine is on the cell represented, then the position of $v$ corresponding to $(a,q)$ where $a$ is the letter in the cell and $q$ the state of the Turing machine, will contain $\critelement$. Otherwise, the position of $v$ corresponding to $(a,\bot)$ will contain $\critelement$. 

As we store the content of a cell in several atoms, we must  ensure that the tape associated with a configuration is consistent, by checking that $\mathbf{v_{next}}$ is consistent with $\mathbf{v}$ from the next cell. To ensure guardedness, we first introduce auxiliary predicates $\data_i^{c}, \data_i^{p}$ and $\data_i^{n}$ that define the content of the $i^\mathrm{th}$ bit of the value of the current, previous and next cells:
\begin{align*}
&\data(c_p,c_n, \mathbf{y^{1,k}}, \mathbf{l_i}(x), \mathbf{v'_{prev}}, \mathbf{v'_{nxt}},r,z,y_0, y_1) \\
&\rightarrow \data_i^{c}(c_n,\mathbf{y^{1,k}},x)
\end{align*}
We now introduce the definition of a mapping $\T_{\kw{data}_n}(x)$ which ensures the consistency of the tape content (note that the first atom is a guard):
\begin{align*}
&\data(c_p,c_n, y_{b_1},\ldots,y_{b_j},y_0,\mathbf{y_1},\mathbf{v}, \mathbf{v_{prev}},\mathbf{l_i}(x),r,z,y_0, y_1) \\
&\wedge \data_i^{c}(c_n, y_{b_1},\ldots,y_{b_j},y_1,\mathbf{y_0},z) 
\end{align*}

$\T_{\kw{data}_p}(x)$ is defined similarly to deal with the previous cell.

We enforce the tape of the initial configuration to have the head of the Turing machine on the first cell (and assume w.l.o.g that this is represented by the first position of $\mathbf{v}$ containing $\critelement$) and all the other cells containing $\flat$ (and we assume w.l.o.g that this is represented by the second position of $\mathbf{v}$ containing $\critelement$). We thus create the mappings $\T_{tape_i}(x)$, for the the first cell, having definition:
\[\data(c_p,c_n, \mathbf{y_0}, l_1(x),\mathbf{v_{prev}},\mathbf{v_{next}},c_p,z,y_0, y_1)\]

and we introduce the mappings $\T_{tape_o}(x)$, for all the other cells, having definition:
\[\data(c_p,c_n, \ldots,y_1,\ldots,a_n, l_2(x),\mathbf{v_{prev}},\mathbf{v_{next}},c_p,z,y_0, y_1)\]

Note that this data is associated with the children of the root (as $p$ is both in the fourth and last minus three positions of the atoms), and not with the root itself, due to the choice of keeping in $\data$ the identifier of the parent of the considered configuration.

  We then check that the tape associated with the $\alpha$-successor of a configuration is indeed obtained by applying an $\alpha$-transition.  This is done by noticing that the value of each cell of the $\alpha$-successor is deterministically defined by the value of the cell and its two neighbors in the original configuration (the neighbors are necessary to know whether the head of the Turing machine is now in the considered cell). To ensure guardedness, we first define a predicate marking $\alpha$-successors (and similarly for $\beta$-successors):
\begin{align*}
& \children_\quantif(c,\alpha,\beta,z,ac_\alpha,ac_\beta,\mathbf{y_0^{k}},\mathbf{y_1^{k}},c_{root},z,y_0,y_1)\\
& \rightarrow {\kw{succ}_\alpha}(c,\alpha)
\end{align*}

Let us consider a cell of address $\mathbf{b^{1,k}}$ in $c_p$. We assume that its content is represented by $i$, while the content of its left (resp. right) neighbor is represented by $j$ (resp. $k$). We represent the fact that this implies that the content of the cell of address $\mathbf{b^{1,k}}$ is $w$ in the $\alpha$-successor of $c_p$ by the following mapping $\T^\alpha_{i,j,k\rightarrow w}(x)$:

\begin{align*}
 &\data(c_p,c_n, \mathbf{b^{1,k}},\mathbf{l_w}(x),\mathbf{v_{prev}}, \mathbf{v_{next}},r,z,y_0, y_1) \\
&\wedge \data_i^{c}(c_p,\mathbf{b^{1,k}},z) \\
&\wedge \data_j^{p}(c_p,\mathbf{b^{1,k}},z) \\
&\wedge \data_k^{n}(c_p,\mathbf{b^{1,k}},z) \\
&\wedge {\kw{succ}_\alpha}(c_p,c_n)
\end{align*}

Note that the above formulation requires the content of the previous and of the next cells, which makes this mappings not applicable when $\mathbf{b^{1,k}}$ is the address of either the first or the last cell. We thus add rules to specifically deal with these two cases (that looks at the content of the current and next cell when $\mathbf{b^{1,k}}$ is a vector of $y_0$, and at the content of current and previous cell when $\mathbf{b^{1,k}}$ is a vector of $y_1$).
Note that there is only polynomially such mappings to be built.
And we finally create a mapping $\T_{accept}(x)$ enforcing that configurations whose tape is in an accepting state (which we assume w.l.o.g. corresponds to the case where the first cell contains the $l^\mathrm{th}$ bit) are declared as accepting. 
\begin{align*}
& \data_l^{c}(c_n,\mathbf{y_0^k},z)\\
 & \wedge  \children_\quantif(c_n,\alpha,\beta,x,ac_\alpha,ac_\beta,\mathbf{y_0^{k}},\mathbf{y_1^{k}},r,z,y_0,y_1)\\
\end{align*}
The policy query is 
\[\children_\exists(root,\alpha,\beta,z,ac_\alpha,ac_\beta,\mathbf{y_0^{1,k}},\mathbf{y_1^{1,k}},root,z,y_0,y_1),\]
We will show that this policy query is disclosed if and only if the original Turing machine accepts on the empty tape.

\subsection{Proof of Correctness}
We show that the policy is disclosed if and only if $\TuringMachine$ accepts on the empty tape. By Theorem
\ref{thm:critinst}, the policy is disclosed if and only if the corresponding $\hocwq$ problem has a positive
answer. Further, this  holds if and only if the policy query holds on the result
of the visible chase  (introduced after Theorem \ref{thm:chase}). We thus focus on showing the equivalence of the acceptance of the empty tape by $\TuringMachine$ and the satisfaction of the policy in the visible chase. 

Let us start by describing some relationships between the visible chase of $\critinst^\targetschema$ and the run of $\TuringMachine$. As $\critinst^\targetschema$ contains $\T_{\init}(\critelement)$, there is in the visible chase the atom
\begin{align*} 
 \children_\exists(c_{root},&c_\alpha,c_\beta,ac_{root},ac_\alpha,ac_\beta,\\
 &\mathbf{y_0^{k}},\mathbf{y_1^{k}},c_{root},\critelement,y_0^0,y_1^0),
\end{align*}

where all individuals but $\critelement$ are nulls. 

\begin{definition}[Tape Representation]
\label{def:tape-representation}
 Let $T$ be a tape (with head position and state included) of $\TuringMachine$. A representation of $T$ is a set of atoms
 \[\{\data(c_p,c_n,\mathbf{a},\mathbf{v}^{\mathbf{a}},\mathbf{v_{prev}}^\mathbf{a},\mathbf{v_{next}}^\mathbf{a},\critelement,y_0,y_1) \}_\mathbf{a},\]
where $\mathbf{a}$ ranges over the binary representations of the addresses of $T$, and such that for any cell of $T$ the following holds:
\begin{itemize}
 \item for any $a$, $\mathbf{v}$ contains fresh nulls except for the 
bit that represents the content of $T$ at address $\mathbf{a}$, where it contains $\critelement$
 \item for any $a$ except the representation of the leftmost 
cell, $\mathbf{v_{prev}}$ contains fresh nulls except for the 
bit that represents the content of $T$ at address $\mathbf{a}-1$;
in this bit it contains $\critelement$ 
($\mathbf{v_{prev}}$ exclusively contains fresh nulls for the leftmost cell)
  \item for any $a$ except the representation of the rightmost cell, 
$\mathbf{v_{next}}$ contains fresh nulls except for the bit that 
represents the content of $T$ at address $\mathbf{a}+1$;
on this bit  it contains $\critelement$ ($\mathbf{v_{next}}$ exclusively contains fresh nulls for the rightmost cell)
\end{itemize}
In that case, $c_n$ is called a representative of $T$.
\end{definition}

\begin{lem}
\label{lem:init}
 $c_\alpha$ and $c_\beta$, as defined above Definition \ref{def:tape-representation}, are representatives of the initial tape.
\end{lem}

\begin{proof} We show the result for $c_\alpha$, the same reasoning being applicable to $c_\beta$.
 As the atom 
 \begin{align*} 
 \children_\exists(c_{root},&c_\alpha,c_\beta,ac_{root},ac_\alpha,ac_\beta,\\
 &\mathbf{y_0^{k}},\mathbf{y_1^{k}},c_{root},\critelement,y_0^0,y_1^0),
\end{align*}
belongs to the visible chase, atoms of the shape
\begin{align*}
 \data(c_{root},c_\alpha,a_1,\ldots,a_{k},\mathbf{v},\mathbf{v_{prev}},\mathbf{v_{next}},c_{root},z,y_0,y_1) 
\end{align*}
for any vector $a_1,\ldots,a_k$ with $a_i \in \{y_0,y_1\}$ for any $i$, are generated, where all nulls from $\mathbf{v},\mathbf{v_{prev}}$ and $\mathbf{v_{next}}$ are fresh (thanks to the rules involving $\genAddr$). As the first argument and the ante-ante-penultimate argument of such an atom are equal, the definition of $\T_{\kw{tape}_i}(x)$ maps to the atom of address $y_0,\ldots,y_0$, and the 
body of $\T_{\kw{tape_o}}(x)$ maps to all the other atoms. Applying $\T_{\kw{data}_n}$ and $\T_{\kw{data}_p}$ then ensures that $c_\alpha$ is a representative for the initial tape, as no other mapping may merge a term 
of these atoms.
\end{proof}

\begin{lem}
\label{lem:children}
 If $c_p$ is a representative of a tape $T$ and if the visible chase contains
 \[\children_\quantif(c_p,\alpha,\beta,ac,ac_\alpha,ac_\beta,\mathbf{y_0^{k}},\mathbf{y_1^{k}},c_{root},\critelement,y_0,y_1),\]
then $\alpha$ (resp. $\beta$) is a representative of the tape $T_\alpha$ (resp. $T_\beta$) obtained by applying the $\alpha$-transition (resp. $\beta$-transition) applicable to $T$.
\end{lem}

\begin{proof}
 We show the result for the $\alpha$-successor, the same reasoning being applicable for the $\beta$-successor. 
As the visible chase contains
\[\children_\quantif(c_p,\alpha,\beta,ac,ac_\alpha,ac_\beta,\mathbf{y_0^{k}},\mathbf{y_1^{k}},c_{root},\critelement,y_0,y_1),\]
it also contains atoms of the shape: 
\begin{align*}
 \data(c_p,\alpha,a_1,\ldots,a_{k},\mathbf{v},\mathbf{v_{prev}},\mathbf{v_{next}},c_{root},\critelement,y_0,y_1),
\end{align*}
for any vector $a_1,\ldots,a_k$, where all nulls from $\mathbf{v},\mathbf{v_{prev}}$ and $\mathbf{v_{next}}$ are fresh. 
Note that $c_p$ is necessary distinct from $c_{root}$ (as it is the representative of a tape).
Hence neither $\T_{\kw{tape}_i}(x)$ nor $\T_{\kw{tape_o}}(x)$ may unify a term with $\critelement$. 
As $c_p$ is a representative of $T$, for any address, if the $i^\mathrm{th}$ bit of $\mathbf{v}$ represents the actual value in $T$ at address $ad$, then the visible chase contains $\data_i^{c}(c_p,\mathbf{b^{1,k}},\critelement)$ where $\mathbf{b^{1,k}}$ is the binary 
encoding of $ad$. Similarly, $\data_i^{n}(c_p,\mathbf{b^{1,k}},\critelement)$ and $\data_i^{p}(c_p,\mathbf{b^{1,k}},\critelement)$ also belong to the visible chase where applicable. 
Then for all addresses, an application of the relevant mapping of the shape $\T^\alpha_{i,j,k\rightarrow w}(x)$ merges the null at the 
position representing the content of $T_\alpha$ with $\critelement$. Applying $\T_{\kw{data}_n}$ and $\T_{\kw{data}_p}$ then ensures that $\alpha$ is a representative for $\T_\alpha$.
\end{proof}

Wrapping up the previous two lemmas, we get that there is a tree structure in the visible chase that corresponds exactly to the tree of 
configurations of the run of $\TuringMachine$: the two 
individuals $c_\alpha$ and $c_\beta$ are representatives of the initial configuration, 
and their children (which are the individuals at the second and third individuals in the $\children$ atom in which they appear 
at the first position) are representatives of the configurations that can be reached with an $\alpha$ or $\beta$ transition. 
It remains to check that the argument representing the accepting status of a configuration are correctly set, which is the topic of the following lemma.

\begin{lem}
\label{lem:acceptance}
 If $c_p$ is the representative of a tape, there is in the visible chase an atom of the shape
 \[\children_\quantif(c_p,\alpha,\beta,
\critelement,ac_\alpha,ac_\beta,\mathbf{y_0^{k}},\mathbf{y_1^{k}},c_{root},\critelement,y_0,y_1),\]

if and only if $\TuringMachine$ accepts on $T$.
\end{lem}

\begin{proof}
 Let $T$ be a tape of representative $c_p$. There are four cases in which $\TuringMachine$ accepts on $T$:
 \begin{itemize}
  \item the state of $\TuringMachine$ is final in $T$: this is the case if and only if $\T_{\kw{accept}}$ merges the fourth argument of $\children_\quantif(c_p,\alpha,\beta,
ac,ac_\alpha,ac_\beta,\mathbf{y_0^{k}},\mathbf{y_1^{k}},c_{root},\critelement,y_0,y_1)$ with $\critelement$
  \item the state of $\TuringMachine$ is universal in $T$ and both its successors are accepting: by induction assumption (on the number of transitions that need to be applied to prove acceptance of a tape), both accepting bits of $\alpha$ and $\beta$ are unified with $\critelement$, and thus the accepting bit of $c_p$ is unified with $\critelement$ thanks to $\T_\forall$
  \item the state of $\TuringMachine$ is existential in $T$ and its $\alpha$-successor is accepting: by induction assumption, the accepting bit of $\alpha$ is unified with $\critelement$, and thus the accepting bit of $c_p$ is unified with $\critelement$ thanks to $\T_{\exists,\alpha}$
  \item similar case, with the $\beta$-successor.
 \end{itemize}

\end{proof}

Let us now use the above lemmas to show that $\TuringMachine$ accepts if and only if the policy query $\secretquery$ holds in the visible chase. If $\TuringMachine$ accepts, let us consider an accepting run of $\TuringMachine$. From Lemmas~\ref{lem:init} and~\ref{lem:children}, we can build a configuration tree that contains a representative for all the tapes that are involved in this run. From Lemma~\ref{lem:acceptance}, the accepting bits of the representative are set adequately, and the policy query holds in the visible chase. 

Conversely, let us consider a visible chase sequence such that the policy query holds in its result. Let us first remark that 
the argument of the policy query appearing in the first position is equal to the argument in the ante-ante-penultimate position. 
This implies that none of the witnesses of mappings other than $\T_\init$ need to be applied in order to entail the policy query, which can be seen from the following three facts: \emph{(i)} none of the positions that may contain a configuration identifier may be unified with $\critelement$; \emph{(ii)} all mappings contain configuration identifiers \emph{(iii)} only the witness associated with $\T_\init$ may generate an atom of the shape
\[\children_\exists(root,\alpha,\beta,z,ac_\alpha,ac_\beta,\mathbf{y_0^{1,k}},\mathbf{y_1^{1,k}},root,z,y_0,y_1),\] 
This implies that in the visible chase sequence entailing the policy query, we start by introducing $\alpha$ and $\beta$ as in Lemma~\ref{lem:init}. Let us now consider the smallest set $S$ of configuration representatives that fulfills the following conditions:
\begin{itemize}
 \item $\alpha$ and $\beta$ are in $S$
 \item if $c$ is in $S$ and the tape associated with $c$ is in a universal state, then both successors of $c$ are in $S$
 \item if $c$ is in $S$ and the tape associated with $c$ is an existential state, then a successor of $c$ having its acceptance bit equal to $\critelement$ is in $S$.
\end{itemize}

By the previous lemmas, there exists an accepting run of~$\TuringMachine$ going exactly through the represented configurations.

\subsection{Second Part of Proof of Theorem \ref{thm:incdguardedlower}: $\exptime$-hardness for Inclusion Dependencies and Guarded Maps in Bounded Arity}

Recall the statement of the second part of Theorem \ref{thm:incdguardedlower}:

\medskip

$\discloseClass(\incd,\guardedmap)$ is $\exptime$-hard even
  in bounded arity.

\medskip

In the proof of Theorem \ref{thm:incdguardedlower}, we used predicates of unbounded arity only to generate exponentially many cell 
addresses. Here, we use only $k$ addresses, and can encode their content through $k$ predicates $\data_1$ to $\data_k$. 
However, the proof follows  the same line of argumentation as in Theorem \ref{thm:incdguardedlower}.

 Let us describe the source signature. 
 \begin{itemize}
  \item $\children_\forall(c,c_\alpha,c_\beta,ac,ac_\alpha,ac_\beta,r,z)$ states that a configuration $c$ is universal and has as children $c_\alpha$ and $c_\beta$, and that the acceptance bit of $c$ is $ac$, of $c_\alpha$ is $ac_\alpha$ and of $c_\beta$ is $ac_\beta$. The last two positions are placeholders for the root of the tree of configurations, and $\critelement$.
  \item $\children_\exists(c,c_\alpha,c_\beta,ac,ac_\alpha,ac_\beta,r,z)$: same meaning, except that $c$ is existential. 
  \item $\data^l(c_p,c_n,\mathbf{v},z)$ states that the cell of address $l$ of the tape represented by $c_n$ has a content represented by $\mathbf{v}$. The last position is a placeholder for $\critelement$.
  \item $\data_i^{l}(c,x,z)$ states that the cell of address $l$ in configuration $c$ contains $x$ at the $i^\mathrm{th}$ position of the representation of its content.
  \item $\kw{succ}_\alpha(c_p,c_n)$ states that $c_n$ is the $\alpha$-successor of $c_p$ (and similarly for $\beta$)
 \end{itemize}

 The symbol $\quantif$ always ranges over $\{\forall,\exists\}$.
 
\paragraph{Initialization.}
We first define a mapping $\T_{\init}(x)$, introducing some elements in the visible chase, whose definition is:
\begin{align*}
 \children_\exists(c_{root},c_\alpha,c_\beta,ac_{root},ac_\alpha,ac_\beta,c_{root},x)
\end{align*}
\paragraph{Generation of the tree of configuration.}
$\alpha$-successors have themselves $\alpha$- and $\beta$-successors, and are existential if their parent is universal:
\begin{align*}
 &\children_\forall(c,\alpha,\beta,ac,ac_\alpha,ac_\beta,r,z) \\
 &\rightarrow \exists \alpha_\alpha,\alpha_\beta,ac_{\alpha_\alpha},ac_{\alpha_\beta}\\
 &\children_\exists(\alpha,\alpha_\alpha,\alpha_\beta,ac_\alpha,ac_{\alpha_\alpha},ac_{\alpha_\beta},r,z)
\end{align*}
And similarly for $\children_\exists$ and for the $\beta$-successor.

\paragraph{Universal and existential acceptance condition.}
If both successors of a universal configuration $n$ are accepting, so is $n$. We create a mapping $\T_\forall(x)$ having definition:
\begin{align*}
 \children_\forall(c,\alpha,\beta,x,z,z,r,z) \\
\end{align*}

If the $\alpha$-successor of an existential configuration $n$ is accepting, so is $n$. We create a mapping $\T_{\exists,\alpha}(x)$ 
having definition:
\begin{align*}
 \children_\exists(c,\alpha,\beta,x,z,ac_\beta,r,z) \\
\end{align*}

We create a similar mapping $\T_{\exists,\beta}$ for the $\beta$-successor.

\paragraph{Tape representation and consistency of tapes.}
We now focus on the representation of the tape and its consistency. For each configuration, we generate $k$ cells whose content is initialized freshly:
\begin{align*}
 &\children_\forall(c,\alpha,\beta,ac,ac_\alpha,ac_\beta,r,z) \\
 &\rightarrow \exists \mathbf{v}~\data^l(c,\alpha,\mathbf{v},z)
\end{align*}
and similarly for existential configurations and for the $\beta$-successors. Note that the values $\mathbf{v}$, $\mathbf{v_{prev}}$ and $\mathbf{v_{next}}$ are again vectors of length the size of $(\letters \cup \{\flat\}) \times (\states \cup \bot)$. We again use the notation $\mathbf{l_i}(x)$ to represent a vector of same length, composed of fresh variables, except for the position $i$, which contains $x$.

To ensure guardedness, we first introduce auxiliary predicates $\data_i^{c}, \data_i^{p}$ and $\data_i^{n}$ that define the content of the $i^\mathrm{th}$ bit of the value of the current, previous and next cells:
\begin{align*}
&\data^l(c_p,c_n,\mathbf{l_i}(x),z) \\
&\rightarrow ~ \data_i^{l}(c_n,x)
\end{align*}
 
We enforce that the tape of the initial configuration has the head of the Turing machine on the first cell (and assume w.l.o.g that this is represented by the first position of $\mathbf{v}$ containing $\critelement$) with all the other cells containing $\flat$.
We also  assume w.l.o.g that the other cells containing $\flat$ is represented by the second position of $\mathbf{v}$ containing $\critelement$. We thus create the mappings $\T_{tape_i}(x)$, for the the first cell, having definition:
\[\data^1(c_p,c_n,l_1(x),c_p)\]

and $\T^l_{tape_o}(x)$, for all the other cells ($2\leq l \leq n$), with definition:
\[\data^l(c_p,c_n,l_2(x),c_p)\]

Note that this data is associated with the children of 
the root (as $c_p$ is both in the first and the penultimate positions of the 
atoms), and not with the root itself, due to the choice of 
keeping in $\data^l$ the identifier of the parent of the considered configuration.

  We then check that the tape associated with the $\alpha$-successor of a configuration is indeed obtained by 
applying an $\alpha$-transition.  
This is done by noticing that the value of each cell of the $\alpha$-successor is determined by the value of the cell 
and its two neighbors in the original configuration (the neighbors are necessary to know whether the head of the Turing machine is now 
in the considered cell). To ensure guardedness, we first define a predicate marking $\alpha$-successors (and similarly for $\beta$-successors):
\begin{align*}
& \children_\quantif(c,\alpha,\beta,z,ac_\alpha,ac_\beta,r,z)\\
& \rightarrow {\kw{succ}_\alpha}(c,\alpha)
\end{align*}

Let us consider a cell of address $l$ in $c_p$. We assume that its content is represented by $i$, while the content of its left (resp. right) neighbor is represented by $j$ (resp. $k$). We represent the fact that this implies that the content of the cell of address $l$ is $w$ in the $\alpha$-successor of $c_p$ by the following mapping $\T^\alpha_{i,j,k\rightarrow w}(x)$:

\begin{align*}
 &\data^l(c_p,c_n,\mathbf{l_w}(x),z) \\
&\wedge~ \data_i^{l}(c_p,z) \\
&\wedge~ \data_j^{l-1}(c_p,z) \\
&\wedge {\kw{succ}_\alpha}(c_p,c_n)
\end{align*}

As in the non-bounded case, the first (resp. last) cell should be dealt with separately, as there is no content in the (non-existent) previous (resp. next) cell. And we finally create a mapping $\T_{accept}(x)$ enforcing that configurations whose tape is in an accepting state (which we assume w.l.o.g. corresponds to the case where the first cell contains the $l_f^\mathrm{th}$ bit) are declared as accepting. 
\begin{align*}
& \data_{l_f}^{1}(c_n,z)\\
 & \wedge  \children_\quantif(c_n,\alpha,\beta,x,ac_\alpha,ac_\beta,r,z)\\
\end{align*}
The policy is 
\[
\children_\exists(root,\alpha,\beta,z,ac_\alpha,ac_\beta,root,z)
\]
We can verify 
that this policy is disclosed if and only if the original Turing machine accepts on the empty tape, using a similar reasoning 
to the unbounded case.

\subsection{Final Part of Proof of Theorem \ref{thm:incdguardedlower}:
  Reduction from $\gtgd$ and $\projectionmap$ to $\incd$ and $\guardedmap$}

Theorem \ref{thm:incdguardedlower} states a $\twoexp$ lower bound for 
general arity and an $\exptime$ lower bound in bounded arity for two
different cases. The first case was when the source constraints
are $\incd$s  and the mappings are guarded. The previous sections
of the appendix have gone through the proofs of this case in detail.
We now finish the proof of Theorem \ref{thm:incdguardedlower}
showing:

\medskip

  $\discloseClass(\gtgd,\projectionmap)$
 is $\twoexp$-hard,
 and is $\exptime$-hard even in bounded arity.

\medskip

The proof of both of these assertions follows directly from
Corollary~\ref{cor:guardedreduce} (the general reduction of maps
presented in section~\ref{ss:simplifymappings}).

We have seen that $\incd$ and $\guardedmap$ reduces to $\gtgd$ and
$\projectionmap$, therefore we have the lower bound for
$\disclose(\gtgd, \projectionmap, \asecretquery)$ from the lower bound
$\disclose(\incd, \guardedmap, \asecretquery)$

\subsection{Proof of Theorem \ref{thm:exptimelowergeneral}: $\exptime$-hardness for Inclusion Dependencies and Atomic Maps,
and for $\ltgd$s with Projection Maps}
Recall the statement of Theorem \ref{thm:exptimelowergeneral}:

\medskip

$\discloseClass(\incd,\atomicmap)$ and $\discloseClass(\ltgd, \projectionmap)$
are both $\exptime$-hard.

\medskip

We first focus on the case of $\discloseClass(\incd,\atomicmap)$.
We adapt the construction used for $\pspace$-hardness of entailment with $\incd$s \cite{casanova} to show  $\Exptime$-hardness 
for $\incd$ source constraints and atomic maps. We start with an alternating (rather than deterministic in \cite{casanova}) 
Turing machine $\M$ and an input $x$, and consider the problem
asking whether there exists a halting computation of $\M$ that uses at most $|x|$ cells. 
As in the original reduction, we use inclusion dependencies to simulate the transition relation of $\M$. The adaptation lies in the 
additional use of a fresh position holding a configuration identifier, and the generation of a tree of configurations,
as in the reduction presented in Theorem \ref{thm:incdguardedlower}.

Let us describe the signature: 
\begin{itemize}
 \item $\config^\quantif(c,ac,\mathbf{v},z)$ states the configuration $c$ has quantification $\quantif$, has accepting bit $ac$, a tape represented by $\mathbf{c}$. the last argument will always hold $\critelement$ in the visible chase;
 \item $\transition^\quantif_{t_\alpha,t_\beta}(c,ac,\mathbf{v},\alpha,ac_\alpha,\beta,ac_\beta,z)$ names two successors configurations $\alpha$ and $\beta$, with the configurations consisting
of acceptance bits $ac_\alpha$ and $ac_\beta$, which are obtained from $c$ by applying transitions $t_\alpha$ and $t_\beta$.
\end{itemize}

Let us turn to the description of $\mathbf{v}$ and subsequently $t_\alpha$. $\mathbf{v}$ represents the content of the tape: 
for each position of the tape, there is an argument for each pair of $\letters \times (\states\cup\{\bot\})$. Intuitively, this argument is equal to $\critelement$ if and only if the position contains the corresponding letter and head, and a fresh null otherwise. 

We introduce a mapping that initializes the tape:
\[\config^\forall(x,ac,\mathbf{v},x)\]

As in the proof of Theorem \ref{thm:incdguardedlower}, we propagate the acceptance information using mappings. 
For a universal state, we use a mapping with definition:
\begin{align*}
 \transition^\forall_{t_\alpha,t_\beta}(c,x,\mathbf{v},\alpha,z,\beta,z,z)
\end{align*}

For an existential state, we use two mappings with definitions:
\begin{align*}
 \transition^\exists_{t_\alpha,t_\beta}(c,x,\mathbf{v},\alpha,z,\beta,ac_\alpha,z)
\end{align*}
and 
\begin{align*}
 \transition^\exists_{t_\alpha,t_\beta}(c,x,\mathbf{v},\alpha,ac_\alpha,\beta,z,z)
\end{align*}

As before, we notice that the state of a cell after applying a transition is deterministically defined by its content as well as the 
content of its left and right neighbor. 
The following inclusion dependency states that from any configuration, we can try to apply all possible transitions to 
generate the $\alpha$- and $\beta$-successors:
\begin{align*}
 &\config^\quantif(c,ac,\mathbf{v},z) \rightarrow \exists \alpha,ac_\alpha,\beta,ac_\beta\\
 &\transition^\quantif_{t_\alpha,t_\beta}(c,ac,\mathbf{v},\alpha,ac_\alpha,\beta,ac_\beta,z)
\end{align*}

We now generate the tape associated with the $\alpha$-transition (and similarly for the $\beta$-transition):
\begin{align*}
 &\transition^\quantif_{t_\alpha,t_\beta}(c,ac,\mathbf{v},\alpha,ac_\alpha,\beta,ac_\beta,z) \rightarrow \\
 &\exists \mathbf{v'_\alpha}~\config^{\vec{\quantif}}(\alpha,ac_\alpha,\mathbf{v_\alpha}\oplus\mathbf{v'_\alpha},z),
\end{align*}
where $\vec{\quantif}$ denotes the dual quantifier. 
Let us describe the vector $\mathbf{v_\alpha}\oplus\mathbf{v'_\alpha}$. Suppose $t_\alpha$ is the transition that checks whether 
position $i$ contains $a$, position $i+1$ contains $b$ and the head 
in state $s$, and position $i+2$ contains $c$; changes $b$ to $b'$, moves the head to the right and goes into state $s'$.
Then $\mathbf{v_\alpha}\oplus\mathbf{v'_\alpha}$ is defined as follows:
\begin{itemize}
 \item any argument that corresponds to a position distinct from $i+1$ or $i+2$ is chosen equal to the argument at the same position in $\mathbf{v}$;
 \item the argument that corresponds to $(i+1,(b',\bot))$ now contains the value of $\mathbf{v}$ at 
position $((i+1),(b,s))$, and all other variables appearing in an argument corresponding to position $(i+1)$ are existentially quantified;
 \item the argument that corresponds to $((i+2),(c,s'))$ now contains the value of $\mathbf{v}$ at position $((i+2),(b,\bot))$, and all other variables appearing in an argument corresponding to position $(i+1)$ are existentially quantified.
\end{itemize}

Note that here we have a distinction with the previous reduction: we do not check that a transition is applicable before applying it, as this would be out of the capabilities of $\incd$. 
However, the same argument as in \cite{casanova} proves that a configuration reached from simulating a non-applicable transition cannot 
lead to an accepting state. We choose as a policy:
\[\config^\forall(x,x,\mathbf{v},x),\]

\begin{proposition} \label{prop:correctexptime}
The policy  is disclosed if and only if there is an accepting computation that uses at most $|x|$ cells.
\end{proposition}

The lower bound for $\discloseClass(\ltgd, \projectionmap)$ follows by reduction:

\begin{proposition} \label{prop:reduceexptimelower}
There is a polynomial time reduction from $\discloseClass(\incd, \atomicmap)$
to $\discloseClass(\ltgd, \projectionmap)$.
\end{proposition}
\begin{proof}
Given a mapping $\phi(\vec x) \rightarrow \exists \vec y ~ H(\vec  t)$ where there may be repeated
variables in the head atom, we replace it by a projection mapping
\[
\phi(\vec x) \rightarrow \exists \vec y ~ H'(\vec  t)
\]
where $H'$ is a new predicate whose  arity is the number of distinct variables in $H(\vec t)$.  $H'(\vec t')$
has the same variables as $H$, but with no repetition. For example, if
the head of the original rule is $H(x,x,y)$, then the new rule has head $H'(x,y)$.

We additionally add the source constraint:
\[
\forall \vec t ~ H'(\vec t) \rightarrow H(\vec t)
\]
It is easy to see that this transformation preserves disclosure.
\end{proof}

\subsection{Proof of Theorem \ref{thm:incdepcqmapboundedlower}: lower bounds for $\incd$s in bounded arity}

Recall the statement of Theorem  \ref{thm:incdepcqmapboundedlower}:

\medskip

 $\discloseClass(\incd,\mappingClass)$
   is $\twoexp$-hard  in bounded arity.

\medskip

\myeat{
The first part of the proof will to show that there is a reduction
from $\discloseClass(\fgtgd,\projectionmap)$ to
$\discloseClass(\incd,\mappingClass)$ but this is given by the
reduction presented in Proposition~\ref{prop:reduceprojmap}: indeed,
since the mappings of $\M'$ are unary they necessarily are composed of
$\fgtgd$.
}

This proof will be very similar to the proof of
Theorem~\ref{thm:incdguardedlower}. We will provide a reduction
from an alternating {\expspace} Turing machine to $\incd$ and
$\singleconstantrule$s. We show how to simulate the run of an
alternating {\expspace} Turing machine $\M$ with inclusion
dependencies and $\singleconstantrule$s.

The main difference between the proof of
Theorem~\ref{thm:incdguardedlower} and the proof here is that in
Theorem~\ref{thm:incdguardedlower} each cell carried $n$ bits
$b_1\dots b_n$ specifying the address of the cell. In this version we
cannot use this trick as we are using a reduction where all predicates
are bounded. For each configuration, the tape will represented in the
leaves of a full binary tree of depth $n$. For a cell $c$, the $n$
bits specifying the address of a $c$ will scattered across the $n$
predicates in its lineage, each holding one bit of the address: an
internal node has two descendants each carrying four values $b,
\vec{b}, y_0, y_1$. We will have $b=y_0$ and $\vec{b}=y_1$ when then
node represents the addresses where the $i$-th bit is $0$ and $b=y_1$,
$\vec{b}=y_0$ when it is $1$.

Let us describe our source signature:
\begin{itemize}
\item
  $\children_\forall(c,c_\alpha,c_\beta,ac,ac_\alpha,ac_\beta,r,y_0,y^{bis}_0,y_1,y^{bis}_1)$
  states that a configuration $c$ is universal and has children
  $c_\alpha$ and $c_\beta$, and that the acceptance bit of $c$ is
  $ac$, of $c_\alpha$ is $ac_\alpha$ and of $c_\beta$ is
  $ac_\beta$. The last four positions are placeholders for $r$ the
  root of the tree of configurations, two values $y_0=y^{bis}_0$
  representing $0$ and two values $y_1=y^{bis}_1$ representing $1$.
\item
  $\children_\exists(c,c_\alpha,c_\beta,ac,ac_\alpha,ac_\beta,r,y_0,y^{bis}_0,y_1,y^{bis}_1)$:
  same meaning, except that $c$ is existential.
\item for $i\in 1..n$, $\addr_i(c_p,c_n,b_{i},\vec{b}_i,y_0,y_1)$
  corresponds to a node of depth $i$ in the binary tree representing
  the tape of a configuration. In this predicate $c_p$ is the parent
  of the node, $c_n$ is the current node, $b_i$ will be equal to $y_0$
  when the node if the first child of $c_p$ and equal to $y_1$
  otherwise. $\vec{b}_i$ will be the complement of $b_i$
  (i.e. $y_0=b_i$ implies $y_1=\vec{i}$ and $y_1=b_i$ implies
  $y_0=\vec{i}$).
\item $\data^c(c,\vec v)$ states that the cell at position $c$
  contains the data represented by $\vec v$. $\data^{p}$ and
  $\data^{n}$ play similar roles for the previous cell and the next
  cell.
\end{itemize}

\paragraph{Critical element.}
We create a mapping $T_{\critelement}(x)$ defined as $\iscrit(x)$.
The relation $\iscrit$ will allow us to test whether a variable is
equal to $\critelement$.

\paragraph{Initialization.}

We first define a mapping $\T_{init}()$ introducing some elements in
the visible chase, whose definition
is: $$\children_{\exists}(c_{root},c_{\alpha},c_{\beta},ac_{root},ac_{\alpha},ac_{\beta},c_{root},y_0,y_0,y_1,y_1)$$

\paragraph{Generating the tree of configuration.}
$\alpha$-successors have themselves $\alpha$- and $\beta$-successors, and are existential if their parent is universal:
\begin{align*}
 &\children_\forall(c,\alpha,\beta,ac,ac_\alpha,ac_\beta,r,y_0,y^b_0,y_1,y_1^b)
  \\ &\rightarrow \exists
  \alpha_\alpha,\alpha_\beta,ac_{\alpha_\alpha},ac_{\alpha_\beta}\\
  &\children_\exists(\alpha,\alpha_\alpha,\alpha_\beta,ac_\alpha,ac_{\alpha_\alpha},ac_{\alpha_\beta},r,y_0,y^b_0,y_1,y_1^b)
\end{align*}
And similarly for $\children_\exists$ and for the $\beta$-successor.

\paragraph{Universal and Existential Acceptance Condition.}
If both successors of a universal configuration $n$ are accepting, so
is $n$. We create a mapping $\T_\forall(x)$ with definition:
\begin{align*}
 \children_\forall(c,\alpha,\beta,x,ac,ac,r,y_0,y^b_0,y_1,y^b_1)  \land \iscrit(ac)\\
\end{align*}

If the $\alpha$-successor of an existential configuration $n$ is
accepting, so is $n$. We create a mapping $\T_{\exists,\alpha}(x)$ of
definition:
\begin{align*}
 \children_\exists(c,\alpha,\beta,x,ac_\alpha,ac_\beta,r,y_0,y^b_0,y_1,y^b_1) \land \iscrit(ac_\alpha)  \\
\end{align*}

We create a similar mapping $\T_{\exists,\beta}$ for the $\beta$-successor.

\paragraph{Generating the tape cells.}
We now focus on the representation of the tape and its consistency. We generate $2^k$ addresses and associated values: 
\begin{align*}
 &\children_{\quantif}(c,\alpha,\beta,ac,ac_\alpha,ac_\beta,r,y_0,y^b_0,y_1,y^b_1)  \\
 &\rightarrow \addr_1(c,c_1,y_0,y_1,y^b_0,y^b_1) 
 \end{align*}
\begin{align*}
 &\children_{\quantif}(c,\alpha,\beta,ac,ac_\alpha,ac_\beta,r,y_0,y^b_0,y_1,y^b_1)  \\
 &\rightarrow \addr_1(c,c_1,y_1,y_0,y^b_0,y^b_1) 
 \end{align*}
And for $i\in 1..n-1$ we have:
\begin{align*}
 &\addr_i(c_{i-1},c_i,b,\vec{b},y^b_0,y^b_1) \\
 &\rightarrow \addr_{i+1}(c_i,c_{i+1},,b,\vec{b},y^b_0,y^b_1) 
 \end{align*}
\begin{align*}
 &\addr_i(c_{i-1},c_i,b,\vec{b},y^b_0,y^b_1) \\
 &\rightarrow \addr_{i+1}(c_i,c_{i+1},\vec{b},b,y^b_0,y^b_1) 
 \end{align*}
Finally for $n$ we have:
\begin{align*}
 &\addr_n(c_{n-1},c_n,b,\vec{b},y^b_0,y^b_1) \\
 &\rightarrow \data^c(c_n,\vec{v})
 \end{align*}

\paragraph{Initialization of the tape.}

For the case 0, we use the pattern $l_1$ and introduce a mapping $\T_{tape0}(x)$ defined as:
\[
\begin{array}{lccr}
 & \children_{\exists}(c_{root},\alpha,\beta,ac,ac_{\alpha},ac_{\beta},y_0,y_0,y_1,y_1) \\
   \land & \addr_1(root,id_1,y_0,y_1,y_0,y_1) \\
   \land&  \dots \\
   \land & \addr_n(id_{n-1},id_n,y_0,y_1,y_0,y_1) \\
   \land & \data^c(id_n,l_1(x)) \\
\end{array}
\]

For all others cases, with a first $1$ at the $i$-th bit, we use the pattern
$l_0$ and introduce $\T_{tape~i}(x)$:
\[
\begin{array}{lccr}
  &\children_{\exists}(c_{root},\alpha,\beta,ac,ac_{\alpha},ac_{\beta},y_0,y_0,y_1,y_1) \\
  \land  & \addr_1(root,id_1,y_0,y_1,y_0,y_1) \\
   \land &  \dots \\
   \land & \addr_{i}(id_{i-1},id_i,y_1,y_0,y_0,y_1) \\
  \land  & \addr_{i+1}(id_{i},id_{i+1},a^i,b^i,y_0,y_1) \\
   \land&  \dots \\
  \land  & \addr_n(id_{n-1},id_n,a^n,b^n,y_0,y_1) \\
  \land  & \data^c(id_n,l_0(x)) \\
\end{array}
\]

\paragraph{Ensuring the coherence between $\data^c$ and $\data^p$.}

We need to check the coherence between $\data^c$ in an address and
$\data^p$ at the previous address. As usual when $v$ is at the address
$\vec{b}10^j$ then $\data^c$ needs to be checked against the $\data^p$
at the address $\vec{b}01^j$. We introduce the mapping
$\T_{prev~j}(x)$:
\[
\begin{array}{llcccc}
  &  \addr_{n-j-1}(id,id_1,y_1,y_0,y_0,y_1) \\
  \land  &  \addr_{n-j-1}(id,id_0,y_0,y_1,y_0,y_1) \\
  \land &  \addr_{n-j}(id_1,id_{10},y_0,y_1,y_0,y_1) \\
  \land  &  \addr_{n-j}(id_{0},id_{01},y_1,y_0,y_0,y_1) \\
    & \dots \\
   \land &  \addr_{n}(id_{10^{j-1}},id_{10^j},y_0,y_1,y_0,y_1) \\
  \land  &  \addr_{n}(id_{01^{j-1}},id_{01^j},y_1,y_0,y_0,y_1)  \\
  \land & \data^c(id_{10^{j}},l_k(\mathbf{s})) \\
  \land & \data^p(id_{10^{j}},l_k(v)) \\
  \land & \iscrit(v) \\

\end{array}  
\]

\paragraph{Encoding transitions.}

As in previous reductions, we encode the transitions of
$\delta_\epsilon$ as a set of $(i,j,k) \rightarrow w$ (where $i$ is
the value of the cell, $j$ is the value at the cell before and $k$
at the cell after and $w$ is the written value).
 
For our transition, we need to write $w$ at the same address s
where $l$ lies in the $\epsilon$-child of the configuration of $c$.
To be at the same address, we need to check that the path follows the
same bits (that we note here $b_n \dots b_1$). We use the mapping
$\T_{i,j,k\rightarrow w}^\epsilon(x)$ defined as: 

\[
\begin{array}{rcl}
   & \iscrit(v) \\
  \land & \data^p(id_{n}, l_i(v)) \\
  \land & \data^c(id_{n}, l_j(c)) \\
  \land & \data^n(id_{n}, l_k(c)) \\
  \land & \data^c(id'_{n}, l_w(x)) \\
  \land &  \addr_n(id_{n},id_{n-1},b_n,\vec{b}_n,y_0,y_1,y_0,y_1) \\
  \land & \addr_n(id'_{n},id'_{n-1},b_n,\vec{b}_n,y_0,y_1) \\
  \land & \addr_{n-1}(id_{n-1},id_{n-2},b_{n-1},\vec{b}_{n-1},y_0,y_1) \\
  \land & \addr_{n-1}(id'_{n-1},id'_{n-2},b_{n-1},\vec{b}_{n-1},y_0,y_1) \\
 & \dots \\
  \land  & \addr_1(id_p,id_1,b_1,\vec{b}_{1},y_0,y_1) \\
  \land & \addr_1(id'_c,id'_{1},\vec{b}_1,b_{1},y_0,y_1) \\
  \land  & Tree_\ell(id_g,id_p,y_0,y_0,y_1,y_1) \\
  \land & \children_{\quantif}(c,c_{\alpha},c_{\beta},ac,ac_{\alpha},ac_{\beta},r,y_0,y_0,y_1,y_1) \\
 \\
\end{array}
\]

\paragraph{Encoding final states.}

Whenever the current state is $q_{accept}$ we need to enforce that the
$ac$ bit is set.  To enforce that the $ac$ bit is set, we introduce the
following mapping, for each value $k\in \{q_{accept}\}\times \Sigma$
marking a final state:

\[
\begin{array}{rcl}
  & \iscrit(v) \\
  \land & \data^c(id_n,l_k(v)) \\
  \land & \addr_n(id_{n-1},id_n,x^{n-1},z^{n-1},y_0,y_1) \\
  & \dots \\
  \land & \addr_1(c,id_1,x^{1},z^{1},y_0,y_1) \\
  \land & \children_{\quantif}(c,\alpha,\beta,x,ac_{\alpha},ac_{\beta},r,y_0,y_0,y_1,y_1) \\
\end{array}
\]

\paragraph{Policy.}

The policy query is 
\[\children_\exists(root,\alpha,\beta,z,ac_\alpha,ac_\beta,r,y_0,y_0,y_1,y_1),\]

We can verify that the policy query is disclosed if and only if the original Turing machine accepts 
on the empty tape.

\subsection{Maximality of our Tractability Conditions}
Recall that Theorem \ref{thm:uidptime}  shows that we can get tractability
by simultaneously restricting our constraints to be $\uid$s and our mappings
to be $\projectionmap$s. Recall also that a $\uid$ is an $\incd$ with at most one
exported variable. Here we show that these restrictions are maximal in the following sense:
if we increase from $\uid$s to $\ltgd$s with frontier one we get intractability.
We also get intractability if we stick with $\uid$s but we allow the mappings to be atomic.
Let $\froneltgd$ denote the $\ltgd$s with at most one exported variable. In fact, we will
show something stronger (here $\emptyset$ denotes no constraints):

\begin{theorem}
$\discloseClass(\emptyset,\atomicmap)$ 
and $\discloseClass(\froneltgd,\projectionmap)$ are both $\np$-hard.
\end{theorem}

In order to prove our results, we will rely again on Proposition \ref{prop:visiblechase},
which states that testing for disclosure is equivalent to evaluating the policy query on the
result of the visible chase process. The process starts with the instance $\hide_\M(\critinst^\targetschema)$,
which has source witnesses for each tuple in $\critinst^\targetschema$. It proceeds by alternating traditional
 chase steps  and merge steps, which are applications of a $\singleconstantrule$.
It is well known that query evaluation is $\np$-hard on arbitrary
instances. But the constraints that we are considering in this section
do not allow us to generate arbitrary instances as a visible sections.
In this section we
will exhibit a instance $\instance$ on which query evaluation is $\np$-hard,
but where $\instance$ can be the result of the visible chase 
using  $\atomicmap$s but no constraints, or with a visible chase using
$\froneltgd$ constraints and $\projectionmap$s.

\paragraph{The instance $\instance$.}
$\instance$ will have one relation $R$ with $6$
atoms. We present the content of $R$ below. Empty cells are filled
with fresh nulls, $c$ is the only value shared by two tuples and $n_i$
correspond to nulls that are shared inside a tuple:

\[
\begin{array}{|c|c|c|c|c|c|c|c|c|}
  \hline
  a & b & \lnot a & \lnot b & a\lor b  & \lnot^2 b & d & \lnot d \\
  \hline
  n_1 & n_1 & c & c & n_1 & & c & n_1 \\
  c & n_2 & n_2 & c & c & & &  \\
  n_3 & c & c & n_3 & c & & &  \\
  c & c & n_4 & n_4 & c & &  & \\
  c &   & n_5 & n_5 &   & c& & \\
    &  & & c  &     & n_6 & n_6 & c \\
  \hline
\end{array}
\]

Note that this is a \emph{single-shared value instance}: only one value, namely
$c$, is shared among multiple tuples.
Such instances can be produced as the result of the visible chase  over
atomic mappings with no constraints. In this case:

\[
\begin{array}{ccc}
R(y,y,x,x,y,v_1,x,y) & \rightarrow & T_1(x) \\
R(x,y,y,x,x,v_1,v_2,v_3) & \rightarrow & T_2(x) \\
R(y,x,x,y,x,v_1,v_2,v_3) & \rightarrow & T_3(x) \\
R(x,x,y,y,x,v_1,v_2,v_3) & \rightarrow & T_4(x) \\
R(x,u,y,y,v_1,x,v_2,v_3) & \rightarrow & T_5(x) \\
R(v_1,v_2,v_3,x,v_4,y,y,x) & \rightarrow & T_6(x) \\
\end{array}
\]

They can also be produced as the result of the visible chase
over one projection mapping $A(x)\rightarrow T(x)$ with $6$ $\froneltgd$s:
\[
\begin{array}{ccc}
A(x)  & \rightarrow & R(y,y,x,x,y,v_1,x,y) \\
A(x) & \rightarrow &  R(x,y,y,x,x,v_1,v_2,v_3) \\
A(x) & \rightarrow & R(y,x,x,y,x,v_1,v_2,v_3) \\
A(x) & \rightarrow & R(x,x,y,y,x,v_1,v_2,v_3) \\
A(x) & \rightarrow & R(x,u,y,y,v_1,x,v_2,v_3)  \\
A(x) & \rightarrow & R(v_1,v_2,v_3,x,v_4,y,y,x)  \\
\end{array}
\]

The remainder of the argument is to show that CQ evaluation is $\np$-hard over this instance,
via reduction from satisfiability of a propositional circuit (Circuit SAT).

\paragraph{General idea of the reduction.}

The reduction that we provide will create a query $Q$ for each
instance $\I$ of Circuit SAT. Without loss of generality, we suppose
that $\I$ is composed of wires $w_1, \dots, w_k$, of negation gates
$N_1,\dots N_l$ and of binary {\org} gates $O_1, \dots, O_m$. Wire that
are not the output of any gate are the inputs of the circuit. We will
suppose that the output corresponds to the wire $1$.

We will build the query $Q$ to contain conjuncts for each wire, each
negation gate and each binary {\org}. Furthermore we will create a
variable $v_i$ for each wire $w_i$.

For the sake of readability, we present the conjuncts graphically, with
each row representing an $R$ atom.
A row with entries $t_{j_1} \ldots t_{j_k}$ represents an  atom
$R(\vec w)$ 
where $w_i$ is a fresh existentially quantified variable
when the cell is empty and the variable $t_{j_i}$ in the cell otherwise.

\paragraph{Wires.}

For each wire $w_i$, we will force the value of its associated
variable $v_i$ to be either $c$ (when the wire carries the value true)
or $n_1$ (when the wire carries false).

For each wire $i$, we will have a conjunct:
\[
\begin{array}{|c|c|c|c|c|c|c|c|c|}
  \hline
  a & b & \lnot a & \lnot b & a\lor b  & \lnot^2 b & d & \lnot d \\
  \hline
  v_i & v_i & & & & & & \\
  \hline
\end{array}
\]

For the variable $v_1$ corresponding to the output wire we also add a
conjunct:

\[
\begin{array}{|c|c|c|c|c|c|c|c|c|}
  \hline
  a & b & \lnot a & \lnot b & a\lor b  & \lnot^2 b & d & \lnot d \\
  \hline
  & &  v_1 & v_1 & &  & & \\
  \hline
\end{array}
\]

\paragraph{Negation.}

For each negation gate $N_k$, whose input is the wire $i$ and 
output is the wire $j$, we will have the following conjuncts:
\[
\begin{array}{|c|c|c|c|c|c|c|c|c|}
  \hline
  a & b & \lnot a & \lnot b & a\lor b  & \lnot^2 b & d & \lnot d \\
  \hline
  v_i &    & r_k   &    &    &    &   &       \\
      &    &      & r_k   &    & p_k &    &    \\
      &    &    &    &    &    & p_k & v_j \\ \hline
\end{array}
\]

\paragraph{Computing binary {\org}.}

For the binary {\org} $O_\ell$ gate whose inputs are the wires $v_i$ and
$v_j$ and the output is $v_k$, we introduce the following conjuncts:

\[
\begin{array}{|c|c|c|c|c|c|c|c|c|}
  \hline
  a & b & \lnot a & \lnot b & a\lor b  & \lnot^2 b & d & \lnot d \\
  \hline
  v_i &    & x_\ell  &    &    &    &   &       \\
      & v_j&    &  y_\ell  &    &    &   &       \\
      &    & x_\ell  &  y_\ell  & v_k & &    &    \\
 \hline
\end{array}
\]

\paragraph{Proof that this reduction captures Circuit-SAT.}

Let us suppose that the circuit is satisfied. Towards showing
that the query is satisfied in the instance $\instance$, we first build a binding
for the variables that are shared between multiple of the conjunct grouping above, which
are exactly the ``wire variables'' $v_i$. We do this by setting $v_i=c$ when
$w_i=\top$ and $v_i=n_1$ when $w_i=\bot$. 
We now show that this binding extends to a valuation making the query true.
Since all the other variables are not shared between the conjunct groups, it suffices
to show satisfiability of each conjunct group in isolation.
\begin{itemize}
\item We see that all the conjuncts corresponding to wires are
  satisfied (even the special conjunct corresponding to the output).
\item For the negation gate $N_k$ whose input is $v_i$ and
output is $v_j$. When $w_i=\top$ and thus $v_i=c$, we can set
$r_k=n_5$, $p_k=c$ and satisfy all 3 conjuncts. When $w_i=\bot$ and thus
$v_i=n_1$, we can set $r_k=c$ and $p_k=n_6$ and satisfy all 3 conjuncts.
\item For an {\org} gate $O_\ell$ whose inputs are $v_i$ and $v_j$, and
  whose output is $v_k$. There are four cases:
\begin{itemize}
  \item when $w_i=w_j=\top$ and thus $v_i=v_j=c$ we can set $x_\ell=y_\ell=n_4$ 
  \item when $w_i=\bot$ and $w_j=\top$ and thus $v_i=n_1$, $v_j=c$ we can set $x_\ell=c$, $y_\ell=n_3$ 
  \item when $w_j=\bot$ and $w_i=\top$ and thus $v_j=n_1$, $v_i=c$ we can set $x_\ell=n_2$, $y_\ell=c$ 
  \item when $w_i=w_j=\bot$ and thus $v_i=v_j=n)1$ we can set $x_\ell=y_\ell=c$ 
\end{itemize}
\end{itemize}
In all cases, our conjuncts are satisfied.

Conversely, let us show that when the query is satisfied in our instance $\instance$, then
the circuit is satisfiable. Let $h$ be a homomorphism from the query
variables to values. Since we have wire conjuncts constraining $v_i$
for each wire $w_i$, we can see that $h(v_i)=n_1$ or $h(v_i)=c$. We
now consider the circuit assignment such that $w_i=\top$ when
$h(v_i)=c$ and $w_i=\bot$ when $h(v_i)=n_1$. Let us show that this
assignment witnesses the satisfiability of the circuit.
\begin{itemize} 
  \item The output wire is already constrained such that $h(v_1)\in
    \{n_1,c\}$ but it also has a special conjunct and the only
    remaining possibility for $h(v_1)$ is $c$ and thus the output gate
    is set at $\top$.
  \item For each negation gate whose input is $w_i$ and output is
    $w_j$:
    \begin{itemize}
    \item when $h(v_i)=n_1$ then the conjunct holding $v_i$ and $r_k$
      (i.e. the first row in the graphical representation) forces that
      $h(r_k)=c$. The conjunct holding $r_k$ and $p_k$ forces $p_k$ to
      be a fresh null or $n_6$. But since $p_k$ appears in the column
      $\lnot^2 b$ and in the column $d$, we can only have $p_k=n_6$
      and thus $v_j=c$.
    \item when $h(v_i)=c$ then the conjunct holding $v_i$ and $r_k$
      forces that $h(r_k)=n_5$ or $h(r_k)=n4$. Then the conjunct
      holding $r_k$ and $p_k$ forces $p_k$ to be either a fresh null
      (when $h(r_k)=n_4$) or $c$ (when $h(r_k)=n_5$). Since $p_k$
      appears in the column $\lnot^2 b$ and in the column $d$ we
      cannot have $p_k$ fresh null, we conclude that$p_k=c$, and thus
      $v_j=n_1$.
    \end{itemize}
    In both cases, the semantics of the negation gate is respected.
  \item Consider each {\org} gate whose inputs are $w_i$, $w_j$ and
    output is $w_k$. First we have:
     \begin{itemize}
    \item when $h(v_i)=n_1$ then necessarily $h(x_\ell)=c$
    \item when $h(v_i)=c$ then $h(x_\ell)=n_2$ or $h(x_\ell)=n_4$ or $h(x_\ell)=n_5$
    \item when $h(v_j)=n_1$ then necessarily $h(y_\ell)=c$
    \item when $h(v_j)=c$ then $h(x_\ell)=n_3$ or $h(x_\ell)=n_4$. 
    \end{itemize} Therefore we see that:
     \begin{itemize}
    \item when $h(v_i)=n_1=h(v_j)$ then necessarily $h(x_\ell)=h(y_\ell)=c$ and thus $h(v_k)=n_1$
    \item when $h(v_i)=c$ and $h(v_j)=n_1$ then $h(x_\ell)=n_2$ and thus $h(v_k)=c$
    \item when $h(v_i)=n_1$ and $h(v_j)=c$ then necessarily $h(y_\ell)=n_3$ and thus $h(v_k)=c$
    \item when $h(v_j)=c=h(v_i)$ then $h(x_\ell)=n_4=h(y_\ell)$ and thus $h(v_k)=c$
     \end{itemize}
     in all cases we do have that the semantics of the {\org} gate is
     respected.
\end{itemize}

All in all, we have seen that the circuit is satisfiable if and only
if the query has a solution on the visible chase.

\subsection{Lower Bounds Inherited from Entailment} \label{app:entailment}

In the body of the paper we claimed that in several cases, we could show that the complexity
of disclosure for a class was at least as hard as the complexity of query entailment for the class.
We do \emph{not} claim that there is a generic reduction  from query entailment to disclosure. 
There is a simple reduction from entailment for special classes of instances to disclosure. More specifically,
disclosure is easily seen to subsume entailment on instances of the form $\hide_\M(\critinst^\targetschema)$.
But one needs to see that entailment on these specialized instances is as hard as entailment in general;
this requires a separate argument for each class.

There are three cases of ``lower bounds from entailment'' that are used in the body of the paper: those
whose lower bound is annotated with $\qentail$ in
Table \ref{tab:summary}. We give the details of each argument below.

\subsubsection{In Bounded Arity, Disclosure with $\incd$ Source Constraints and Projection Maps is $\np$-hard}
We begin by showing that disclosure for $\incd$ source constraints and projection maps inherits the $\np$-hardness
that is known for query entailment with $\incd$s. We do this via a direct reduction from $3$-coloring.
We make use again of the characterization of disclosure using the visible chase.

Let us take a graph $G=(V,E)$ that is an input to $3$-coloring.
In our reduction, the
schema, the constraints and the mapping will not depend on this actual
graph reduced. Only the query will depend on the graph.

We will have a single source relation $OK(x,y,z)$ and one mapping $OK(x,y,z) \rightarrow M()$ 
to create canonical values for $(x_0,y_0,z_0)$ in $\hide_\M(\critinst^\targetschema)$, which
is the initial instance in the visible chase. Then we will use
two $\incd$ constraints to create all permutations for these values:
$OK(x,y,z) \rightarrow OK(x,z,y)$ and $OK(x,y,z) \rightarrow
OK(y,x,z)$.

Because of the mapping, $\hide_\M(\critinst^\targetschema)$ will have three values
$x_0,y_0,z_0$ with $OK(x_0,y_0,z_0)$. Then the constraints ensure that the
canonical model contains the six permutations of arguments for $OK$:
$OK(x_0,y_0,z_0)$, $OK(x_0,z_0,y_0)$, $OK(y_0,x_0,z_0)$, $OK(y_0,z_0,x_0)$, $OK(z_0,x_0,y_0)$, $OK(z_0,y_0,z_0)$.

In our query $|V|$ variables will capture the coloring of each node,
we note $v(n)$ the variable associated with node $n$.  For each
$(f,t)\in E$ the query will include a conjunct $\exists c  ~~ OK(v(f),v(t),c)$.

We sketch the correctness of this reduction.
The three values $x_0,y_0,z_0$ in $\hide_\M(\critinst^\targetschema)$  encode the three
possible colors in a coloring. The conjunct $\exists c ~~
OK(v(f),v(t),c)$  forbids the nodes $f$ and $t$ to be mapped
to the same value ($x_0$, $y_0$ or $z_0$) as $\exists c ~~OK(v,v,c)$ has no
solution in $\hide_\M(\critinst^\targetschema)$. Therefore if we have disclosure
have a $3$-coloring.

Conversely, if we have a $3$-coloring, we can find a solution for the
query in the visible chase.

\subsubsection{In General Arity, Disclosure with $\incd$ Source Constraints and Projection Maps is $\pspace$-hard}

Here we give a direct reduction\footnote{The original proof given here
  was faulty, many thanks to Balder ten Cate for noticing it and
  suggesting a fix.}  from the implication problem for $\incd$s, or
equivalently, the entailment problem for a single-atom instance and an
atomic query. This is known to be $\pspace$-hard \cite{casanova}.
Given a problem $\Sigma \vDash R_1(\vec{X}) \subseteq R_2(\vec{X})$
(where $\vec{X}$ has no repeated variables and $\Sigma$ is composed of
IDs), we introduce a fresh predicate $shadowR_1$ and we reduce it to
the disclosure problem with query $shadowR_1(\vec{X})\land
R_2(\vec{X})$ on the constraints $\Sigma$ plus
$shadowR_1(\vec{X})\rightarrow R_1(\vec{X})$ and the mapping $V() :=
\exists \vec{X} ~~ shadowR_1(\vec{X})$.

\subsubsection{In Bounded Arity, Disclosure with $\fgtgd$ Source Constraints and Projection Maps is $\twoexp$-hard}

The last  place where we claim that disclosure is at least as hard as entailment is for $\fgtgd$ source constraints
and projection maps in bounded arity. Here we will proceed by modifying the reduction used in 
 Theorem  \ref{thm:incdepcqmapboundedlower}. In this proof we used mappings for two distinct purposes. 
The initialization mapping $\T_{init}()$ was used to generate some values in the initial instance
of the visible chase. In the proof, this mapping is an atomic map but not a projection map; but we can easily  change
this to use a projection map and an $\ltgd$. 

The remaining maps are used to ensure that certain values get merged with $\critelement$ in the visible chase.
Put another way, they are used to enforce certain $\singleconstantrule$s. But with the mappings $\T_{init}()$
and  $T_{\critelement}(x)$, we can ensure that the initial instance of the visible chase includes  exactly  one
element satisfying $\iscrit$. Once we have done this, we can mimic a $\singleconstantrule$
\[
\phi(\vec x) \rightarrow x_i=\critelement
\]
by a source constraint
\[
\phi(\vec x) \rightarrow \iscrit(x_i)
\]
This must be a $\fgtgd$, since the frontier has size one.
Transforming the mappings according to this methodology, while leaving the query the same  as in
 Theorem  \ref{thm:incdepcqmapboundedlower} gives us a modification of the hardness proof  using
$\fgtgd$ source constraints and projection maps, as required.

\newpage
\section{Refinements of our results} \label{sec:appglobal}
\paragraph{Atomic queries.}
We have focused in the body of the paper on policy queries given as general CQs. But
almost all of our lower bounds can be seen to hold for atomic queries. The only exceptions
are stated in Theorem \ref{thm:linlinupper}
and Corollary \ref{cor:ididupper}, where we claim $\ptime$ membership in bounded arity when restricting
to atomic queries.
 Note that the  $\np$-hardness bounds for general CQs corresponding to these upper
bounds do not follow from our custom reductions,
but using the simple reduction from entailment of CQs for the corresponding classes (e.g. $\incd$s).

\paragraph{Non-Boolean queries.}
In this appendix we have provided details of our upper bounds, assuming for simplicity that
the queries $\asecretquery$ are Boolean. But the proofs all extend to the non-Boolean case, as we now explain.
To see this we need to go back to  Theorem  \ref{thm:critinst}. We restate the theorem
in a slightly different variant:

\begin{theorem} \cite{bbcplics}
When source constraints are TGDs and mapping rules are given by CQ definitions, then if a disclosure of a CQ
(Boolean or non-Boolean) occurs,
then the source instance which witnesses this can be taken to be $\critinst^\sourceschema$.
\end{theorem}
The statement  differs slightly from that of Theorem  \ref{thm:critinst}, since this version talks about
getting an instance that agrees with  $\critinst^\sourceschema$ on the mapping images, rather than
having one that extends $\hide_\M(\critinst^\targetschema)$ and satisfies the constraints.

The important point is that the result holds for non-Boolean queries as well as Boolean queries.
Note that  for a Non-Boolean query $\asecretquery(\vec x)$, all the facts that an attacker
will see  in the mapping image of $\critinst^\sourceschema$ will contain only the value $\critelement$. Thus
the only query answers that can be disclosed to the attacker will involve the value
$\critelement$. Inspection of each of the upper bound reductions
will show  that to detect such disclosures, it suffices to
pre-process the query to add conjuncts $\iscrit(x_i)$ for each variable $x_i$, treating the result as a Boolean query.

Note that this transformation converts atomic queries to queries consisting of a single atom and an additional set of unary
atoms. However, this will not impact the $\ptime$ claims in Theorem \ref{thm:linlinupper}
and Corollary \ref{cor:ididupper}. For example in  Theorem \ref{thm:linlinupper}, we will
need only to note that for atomic queries on a bounded arity  schema, we will get an atomic
query with a bounded number
of additional atoms of the form $\iscrit(x_i)$. Entailment of such queries over a bounded arity schema with
$\incd$s is still in $\ptime$.

\paragraph{Dependencies with multiple atoms in the head.}
In some of our upper bound proofs, we assumed that the dependencies had a single atom in the head for simplicity, even
when the classes in question (e.g. $\gtgd$s) does not impose this.
In our results that are stated for general arity, this assumption can be made without loss of generality, since
one can simplify the heads by introducing intermediate predicates.  In bounded arity, one must
take some care, since one cannot polynomially reduce to the case of a single atom in the head.
All of our results for bounded arity do in fact hold as stated, without any additional restrictions on the head.
We explain how the argument needs to be customized for the most subtle case, Theorem \ref{thm:linlinupper}.

Recall that the bounded arity case of Theorem \ref{thm:linlinupper}  starts with the critical-instance rewriting, which reduces to reasoning
with Guarded TGDs having a fixed side signature, the unary predicate $\iscrit(x)$.
The linearization of \cite{antoinemichael,antoinemichaelarxiv}, applied in this context,  proceeds in two steps. First we generate all
derived rules of the form:
\[
R(\vec x) \wedge \bigwedge_i \iscrit(x_i) \rightarrow \iscrit(x_j)
\]
Notice that these are \emph{full-dependencies}: no existentials in the head.
This generation can be done inductively, via the dynamic programming steps  in  \cite{antoinemichaelarxiv}: one inductive steps
composes a derived rule with one of the original non-full dependencies. A second step  composes two derived  full rules. This can be applied
directly to the case of rules with multiple atoms in the head.

After this is done, the second step of linearization moves to an  extended signature described as follows: for every relation $R$ of arity $k$ in the original signature
(without $\iscrit$), and for each set of positions $P$ of $R$, we introduce predicates $R^P$ of arity  $k$. Informally,
$R^P(\vec x)$ stands in for $R(\vec x) \wedge \bigwedge_{i \in P} \iscrit(x_i)$.
We lift every original dependency:
\[
R(\vec x) \wedge \bigwedge_{i \in P} \iscrit(x_i) \rightarrow \exists \vec y ~ \bigwedge H_j(\vec t_j)
\]
to a linear TGD:
\[
\bigwedge R^P(\vec x) \rightarrow \exists \vec y ~ \bigwedge H^{P_j}_j(\vec t_j)
\]
where $P_j$ contains the positions corresponding to exported variables in $P$.
We lift every derived full dependency of the form:
\[
R(\vec x) \wedge \bigwedge_{i \in P} \iscrit(x_i) \rightarrow \iscrit(x_j)
\]
to a linear TGD:
\[
R^P(\vec x) \rightarrow R^{P \cup \{j\}}(\vec x)
\]
Finally we have linear TGD asserting that the semantics of $R^P$ become stronger as one adds to the set of positions $P$:
\[
R^{P'}(\vec x) \rightarrow R^{P} (\vec x)
\]
for $P \subset P'$.

We rewrite the query to the extended signature in the analogous way.
The correctness of this transformation is given by an argument identical to that in the single-headed case in  \cite{antoinemichaelarxiv}.

Note that this  transformation is in $\ptime$ when the arity is fixed.
It reduces us to an entailment problem with $\ltgd$s, still with bounded arity, but with multiple atoms in the head.
Such an entailment problem  can be shown to be in $\np$ using a simple variation of the algorithm
for $\incd$s of \cite{johnsonklug}.

\end{document}